\documentclass[12pt,notitlepage]{amsart}%
\usepackage{amssymb}
\usepackage{amsfonts}
\usepackage{graphicx}
\usepackage{amscd}
\usepackage{graphicx}
\usepackage{amsmath}%
\newtheorem{theorem}{Theorem}
\theoremstyle{plain}

\newtheorem{definition}{Definition}

\newtheorem{lemma}{Lemma}
\newtheorem{notation}{Notation}

\newtheorem{proposition}{Proposition}
\newtheorem{remark}{Remark}

\numberwithin{equation}{section}
\numberwithin{theorem}{section}
\numberwithin{lemma}{section}
\numberwithin{proposition}{section}
\numberwithin{corollary}{section}

\ifx\pdfoutput\relax\let\pdfoutput=\undefined\fi
\newcount\msipdfoutput
\ifx\pdfoutput\undefined\else
\ifcase\pdfoutput\else
\ifx\paperwidth\undefined\else
\ifdim\paperheight=0pt\relax\else\pdfpageheight\paperheight\fi
\ifdim\paperwidth=0pt\relax\else\pdfpagewidth\paperwidth\fi
\fi\fi\fi
\begin{document}
\title[Ultrametric Transition Networks]{Ultrametric Diffusion, Rugged Energy Landscapes and Transition Networks}
\author[Z\'{u}\~{n}iga-Galindo]{W. A. Z\'{u}\~{n}iga-Galindo}
\address{University of Texas Rio Grande Valley\\
School of Mathematical \& Statistical Sciences\\
One West University Blvd\\
Brownsville, TX 78520, United States}
\email{wilson.zunigagalindo@utrgv.edu}
\thanks{The author was partially supported by the Lokenath Debnath Endowed
Professorship, UTRGV}

\begin{abstract}
In this article we introduce the ultrametric networks which are $p$-adic
continuous analogues of the standard Markov state models constructed using
master equations. A $p$-adic transition network (or an ultrametric network) is
a model of a complex system consisting of a hierarchical energy landscape, a
Markov process on the energy landscape, and a master equation. The energy
landscape consists of a finite number of basins. Each basin is formed by
infinitely many network configurations organized hierarchically in an infinite
regular tree. The transitions between the basins are determined by a
transition density matrix, whose entries are functions defined on the energy
landscape. The Markov process in the energy landscape encodes the temporal
evolution of the network as random transitions between configurations from the
energy landscape. The master equation describes the time evolution of the
density of the configurations. We focus on networks where the transition rates
between two different basins are constant functions, and the jumping process
inside of each basin is controlled by a $p$-adic radial function. We
solve\ explicitly the Cauchy problem for the master equation attached to this
type of networks. The solution of this problem is the network response to a
given initial concentration. If the Markov process attached to the network is
conservative, the long term response of the network is controlled by a Markov
chain. If the process is not conservative the network has absorbing states. We
define an absorbing time, which depends on the initial concentration,\ if this
time is finite the network reaches an absorbing state in a finite time. We
identify in the response of the network the terms responsible for bringing the
network to an absorbing state, we call them the fast transition modes. The
existence of the fast transition modes is a consequence of the assumption that
the energy landscape is ultrametric (hierarchical), and to the best of our
understanding this result cannot be obtained using standard methods of Markov
state models. Nowadays, it is widely accepted that protein native states are
kinetic hubs that can be reached quickly from any other state. The existence
of fast transition modes implies that certain states on an ultrametric network
work as kinetic hubs.

\end{abstract}
\maketitle

\section{Introduction}

A central paradigm in physics of complex systems (such proteins or glasses)
asserts that the dynamics of such systems can be modeled as a random walk in
the energy landscape of the system, see \cite{Fraunfelder et al}%
-\cite{Fraunfelder et al 3}, \cite{Wales}, see also \cite{KKZuniga},
\cite{Kozyrev SV}, and the references therein. By a complex system we
understand one that can assume a very large number of non-identical
conformations (states) that are organized hierarchically. Fluctuations and
relaxations correspond to jumps between states. The hierarchical structure of
the energy landscapes of the mentioned systems was discovered in the 80s by
Parisi et al in the context of the spin glass theory \cite{M-P-V}, see also
\cite{R-T-V}, and by Frauenfelder in the context of protein physics
\cite{Fraunfelder et al}.

These landscapes have a huge number of local minima. It is clear that a
description of the dynamics on such landscapes require an adequate
approximation. The interbasin kinetics offers an acceptable solution to this
problem. By using this approach an energy landscape is approximated by an
ultrametric space (a rooted tree) and a function on this space describing the
distribution of the activation barriers, see e.g. \cite{Becker et al},
\cite{Stillinger et al 1}-\cite{Stillinger et al 2}.

The dynamics is controlled by a master equation of the form%
\begin{equation}
\frac{d}{dt}p_{I}\left(  t\right)  =\sum_{J}\left\{  w_{I,J}p_{J}\left(
t\right)  -w_{J,I}p_{I}\left(  t\right)  \right\}  , \label{1}%
\end{equation}
where $I$, $J$ runs through the leaves of the disconnectivity graph (a rooted
tree), $w_{I,J}\geq0$ is the transition rate peer unit of time from state $J$
to state $I$, and $p_{I}\left(  t\right)  $\ is the probability of finding the
system at state $I$ at time $t$, see e.g. \cite{KKZuniga}, \cite{Kozyrev SV},
and the references therein. In practical applications, the matrices $\left[
w_{I,J}\right]  $ are very large, consequently, the study of the qualitative
behavior of complex systems via the master equation approach is extremely difficult.

A finite rooted tree is the most basic example of an ultrametric space. A
\ metric space $(M,d)$ is called ultrametric if the metric $d$ satisfies
$d(a,b)\leq\max\left\{  d(a,b),d(c,b)\right\}  $ for any three points $a$,
$b$, $c$ in $M$. The ultrametric spaces constitute the natural language to
formulate models of phenomena where hierarchy plays a central role.

In this article we introduce a new $p$-adic exactly solvable models for the
dynamics of certain complex systems, and provide an application to protein
folding. The field of $p$-adic numbers $\mathbb{Q}_{p}$ has a paramount role
in the category of ultrametric spaces. From now on $p$ denotes a fixed prime
number. A $p$-adic number is a series of the form%
\begin{equation}
x=x_{-k}p^{-k}+x_{-k+1}p^{-k+1}+\ldots+x_{0}+x_{1}p+\ldots,\text{ with }%
x_{-k}\neq0\text{,}\label{p-adic-number}%
\end{equation}
where the $x_{j}$s \ are $p$-adic digits, i.e. numbers in the set $\left\{
0,1,\ldots,p-1\right\}  $. The set of all possible series of form
(\ref{p-adic-number}) constitutes the field of $p$-adic numbers $\mathbb{Q}%
_{p}$. There are natural field operations, sum and multiplication, on series
of form (\ref{p-adic-number}), see e.g. \cite{Koblitz}. There is also a
natural norm in $\mathbb{Q}_{p}$ defined as $\left\vert x\right\vert
_{p}=p^{k}$, for a nonzero $p$-adic number $x$ of the form
(\ref{p-adic-number}). The field of $p$-adic numbers with the distance induced
by $\left\vert \cdot\right\vert _{p}$ is a complete ultrametric space. The
ultrametric property refers to the fact that $\left\vert x-y\right\vert
_{p}\leq\max\left\{  \left\vert x-z\right\vert _{p},\left\vert z-y\right\vert
_{p}\right\}  $ for any $x$, $y$, $z$ in $\mathbb{Q}_{p}$. The field of
$p$-adic numbers has a fractal structure, see e.g. \cite{Alberio et al},
\cite{V-V-Z}.

The construction of a mathematical theory for master equations of type
(\ref{1}) over arbitrary finite ultrametric spaces (for instance general
graphs) is a natural step in understanding the dynamics of complex systems. In
our view there are two approaches to this problem. The first one is to study
master equations\ on arbitrary graphs as discrete objects. The second one is
to construct `continuous versions' of master equations of type (\ref{1}) over
arbitrary graphs. The first approach is based almost exclusively on knowing
the spectra of the matrices $\left[  w_{I,J}\right]  $. Since these matrices
are typically very large, this approach is useful only in computer simulations.

The second approach is based on the fact that matrices $\left[  w_{I,J}%
\right]  $ can be realized as $p$-adic operators acting on suitable functions
spaces. We discuss this construction in the Section \ref{Section_4}. The
$p$-adic continuous versions of the master equations of form (\ref{1}) are
equations of type%
\begin{equation}
\frac{du\left(  x,t\right)  }{dt}=\int\limits_{\mathcal{K}}\left\{  j(x\mid
y)u(y,t)-j(y\mid x)u(x,t)\right\}  dy, \label{3}%
\end{equation}
where $\mathcal{K}$ is an open compact subset of $\mathbb{Q}_{p}$,
$x\in\mathcal{K}$, $t\geq0$, $u\left(  x,t\right)  $ is the population (or
density) of state $x$ at time $t$, and the functions $j(x\mid y)$,
respectively $j(y\mid x)$, give the transition density rate (per unit of time)
from $y$ to $x$, respectively\ from $x$ to $y$, and $dy$ is the normalized
Haar measure of the group $(\mathbb{Q}_{p},+)$.

For the sake of simplicity, we take $\mathcal{K}=%
{\textstyle\bigsqcup\nolimits_{a\in\mathcal{G}}}
\left(  a+p\mathbb{Z}_{p}\right)  $, where $\mathcal{G}\subset\left\{
0,1,\ldots,p-1\right\}  $. Each ball $a+p\mathbb{Z}_{p}$ corresponds to a
basin, and the states belonging to $a+p\mathbb{Z}_{p}$ are organized in an
infinite tree with root $a$. Under mild hypotheses, \ all the equations of the
form (\ref{3}) are $p$-adic diffusion equations, which means that there is a
Markov process with state space $\mathcal{K}$ attached to each of these
equations, see Theorem \ref{Theorem1}. The dynamics complex system attached to
(\ref{3}) can be described as random transitions in the state space
$\mathcal{K}$. The fact that the dynamics of certain complex systems can be
modeled using $p$-adic reaction diffusion equations in $\mathbb{Q}_{p}$ was
presented in \cite{Av-4}, see also \cite{Av-5}, \cite{KKZuniga},
\cite{Zuniga-LNM-2016}, \cite{Zuniga-Nonlinear}, \cite{Zuniga-Galindo-JMAA},
\ and \ the references therein. Indeed, an equation of type (\ref{3}%
),\textit{\ but with }$\mathbb{Q}_{p}$\textit{\ instead of }$\mathcal{K}$ was
proposed in \cite{Av-4}. The dynamics of this type of equations is radically
different to the dynamics of (\ref{3}).

The function $u\left(  x,t\right)  :\mathcal{K}\times\left[  0,\infty\right)
\rightarrow\mathbb{R}$ is completely determined \ by the restrictions
$u_{a}\left(  x,t\right)  =u\left(  x,t\right)  $ for $x\in a+p\mathbb{Z}_{p}%
$, $t\geq0$, i.e. $u\left(  x,t\right)  =\sum_{a\in\mathcal{G}}u_{a}\left(
x,t\right)  $. Similarly, $j(x\mid y)=\sum_{a\in\mathcal{G}}\sum
_{b\in\mathcal{G}}j_{a,b}(x\mid y)$, where $j_{a,b}(x\mid y):\left(
a+p\mathbb{Z}_{p}\right)  \times\left(  b+p\mathbb{Z}_{p}\right)
\rightarrow\mathbb{R}$. It is more convenient to perform the change of
variables $x\rightarrow a+x$, and work with the functions $\widetilde{u}%
_{a}\left(  x,t\right)  =u_{a}\left(  a+x,t\right)  $ with $x\in
p\mathbb{Z}_{p}$, and $\widetilde{j}_{a,b}(x\mid y)=j_{a,b}(a+x\mid b+y)$ with
$x,y\in p\mathbb{Z}_{p}$. With this notation equation (\ref{3}) can be written
as%
\begin{equation}
\frac{d\widetilde{u}_{a}(x,t)}{dt}=\sum\limits_{b\in\mathcal{G}}\text{
}\left\{  \text{ }\int\limits_{p\mathbb{Z}_{p}}\left\{  \widetilde{u}%
_{b}(y,t)-\widetilde{u}_{a}(x,t)\right\}  \widetilde{j}_{a,b}(x\mid
y)dy\right\}  -\widetilde{u}_{a}(x,t)\widetilde{S}_{a}(x)\text{ for }%
a\in\mathcal{G}\text{.} \label{4A}%
\end{equation}
This systems of integro-differential equations is the master equation of a
$p$\textit{-adic transition network (or simply a network). }Notice that we can
attach to equation (\ref{4A}) an oriented graph whose vertices are $\left\{
a_{1},,\ldots,a_{\#\mathcal{G}}\right\}  $. If $\widetilde{j}_{a,b}(x\mid
y)\neq0$, there is a directed edge from $b$ to $a$. In this article, a network
is a model of a complex system consisting of three components: an energy
landscape, a Markov process on the energy landscape, and a master equation.
The energy landscape consists of a finite number of \textit{basins}
$a+p\mathbb{Z}_{p}$, $a\in\mathcal{G}$. The transitions between the basins are
determined by \textit{the transition density matrix of the network} as
$\left[  \widetilde{j}_{a,b}(x\mid y)\right]  _{a,b\in\mathcal{G}}$, and the
\textit{sink function} as $\left[  \widetilde{S}_{a}(x)\right]  _{a\in
\mathcal{G}}$.

The solution of the Cauchy problem attached to (\ref{4A}) is the key tool to
understand the qualitative behavior of the $p$-adic transition networks. To
obtain precise description of the solutions $u(x,t)$, it is necessary to
impose restrictions to the transition functions $j_{a,b}(x\mid y)$,
$j_{b,a}(y\mid x)$. We assume that the restriction of the functions $j(x\mid
y)$, $j(y\mid x)$ to the set $\left(  a+p\mathbb{Z}_{p}\right)  \times\left(
a+p\mathbb{Z}_{p}\right)  $ have the form $j(x\mid y)=w_{a}(\left\vert
x-y\right\vert _{p})$ and \ $j(y\mid x)=v_{a}(\left\vert x-y\right\vert _{p}%
)$, and that the restriction of the functions $j(x\mid y)$, $j(y\mid x)$ to
the set $\left(  a+p\mathbb{Z}_{p}\right)  \times\left(  b+p\mathbb{Z}%
_{p}\right)  $, with $a\neq b$, have the form $j_{a,b}(x\mid y)=\lambda
_{a,b}\Omega\left(  p\left\vert x-a\right\vert _{p}\right)  \Omega\left(
p\left\vert y-b\right\vert _{p}\right)  $, $j_{b,a}(y\mid x)=\mu_{b,a}%
\Omega\left(  p\left\vert x-a\right\vert _{p}\right)  \Omega\left(
p\left\vert y-b\right\vert _{p}\right)  $, where $\Omega\left(  p\left\vert
x-a\right\vert _{p}\right)  $ denotes the characteristic function of the ball
$a+p\mathbb{Z}_{p}$. This choice means that we assume that the transition
rates between different basins are constant functions, while the jumping
process inside of the ball $a+p\mathbb{Z}_{p}$ is controlled by an operator of
the form
\[
\boldsymbol{W}_{a}\varphi\left(  x\right)  =\int\limits_{p\mathbb{Z}_{p}%
}\left\{  \varphi(y,t)-\varphi(x,t)\right\}  w_{a}(\left\vert x-y\right\vert
_{p})dy,
\]
where $\varphi$ is a function supported in the ball $p\mathbb{Z}_{p}$. These
operators are pseudo-differential and their spectra can be computed
explicitly:
\[
\boldsymbol{W}_{a}\Psi_{rmj}\left(  x\right)  =\left(  \widehat{w_{a}}%
(p^{1-r})-\gamma_{a}\right)  \Psi_{rmj}\left(  x\right)  ,
\]
where the radial function $\widehat{w_{a}}$ is the Fourier transform of
$w_{a}$, $\gamma_{a}$ is a positive constant \ and the functions $\left\{
\Psi_{rmj}\right\}  $ are an orthonormal basis of $L^{2}(p\mathbb{Z}_{p})$.

The solution of the Cauchy problem is%
\begin{equation}
u_{a}\left(  x,t\right)  =p^{\frac{1}{2}}C_{0}^{\left(  a\right)  }\left(
t\right)  \Omega\left(  p\left\vert x-a\right\vert _{p}\right)  +%
{\textstyle\sum\limits_{rmj}}
e^{\left(  \widehat{w}_{a}\left(  p^{1-r}\right)  -p^{-1}\overline{\mu}%
_{a}\right)  t}\operatorname{Re}\left(  C_{rmj}^{\left(  a\right)  }\left(
0\right)  \Psi_{rmj}\left(  x\right)  \right)  , \label{Solution1}%
\end{equation}
$a\in\mathcal{G}$, where $\left[  C_{0}^{\left(  a\right)  }\left(  t\right)
\right]  _{a\in\mathcal{G}}=e^{\Lambda t}\left[  C_{0}^{\left(  a\right)
}\left(  0\right)  \right]  _{a\in\mathcal{G}}$. The constants $C_{0}^{\left(
a\right)  }\left(  0\right)  $ are real, while the constants $C_{rmj}^{\left(
a\right)  }\left(  0\right)  $ are complex . Here $\Lambda$ is a matrix
determined by the $\lambda_{a,b}$ and the $\mu_{b,a}$, and $\widehat{w}%
_{a}\left(  p^{1-r}\right)  -p^{-1}\overline{\mu}_{a}\leq0$ for any
$a\in\mathcal{G}$.

The main results of this work are the theorems \ref{Theorem2}, \ref{Theorem3},
and \ref{Theorem4}. Here we provide an overview of the meaning of these
results. If the Markov process attached to the network is conservative, \ then
for any continuous initial datum $u^{0}(x)\in\left[  0,1\right]  $, the
solution of the Cauchy problem $u(x,t)\in\left[  0,1\right]  $ can be expanded
as (\ref{Solution1}) for $t\in\left[  0,\infty\right)  $, cf. Theorem
\ref{Theorem2}. The matrix $\Lambda$ satisfies $\Lambda\boldsymbol{1}%
=\boldsymbol{0}$ or $\Lambda\boldsymbol{1}\neq\boldsymbol{0}$, cf. Theorem
\ref{Theorem3}. In the first case $\left\{  e^{t\Lambda}\right\}  _{t\geq0}$
is a Markov semigroup of matrices and $\Lambda$ is the infinitesimal generator
of a time-continuous Markov chain with states $\#\mathcal{G}$. In the second
case, $\left\{  e^{t\Lambda}\right\}  _{t\geq0}$ is a substochastic semigroup
of matrices, and by introducing a terminal state $\Delta$, $\Lambda$ becomes
the infinitesimal generator of a time-continuous Markov chain with states
$1+\#\mathcal{G}$. The long-term behavior of the network is controlled by the
Markov chain attached to $\Lambda$, see Theorem \ref{Theorem4}.

If the Markov process attached to the network is not conservative, \ the
network may have absorbing states. Given an continuous initial datum
$u^{0}(x)\in\left[  0,1\right]  $, we introduce an absorbing time $\tau
=\tau\left(  u^{0}\right)  $. Then solution of the Cauchy problem
$u(x,t)\in\left[  0,1\right]  $ can be expanded as (\ref{Solution1}) for
$t\in\left[  0,\tau\right]  $, cf. Theorem \ref{Theorem2}. If $\tau=\infty$,
the network does not have absorbing states and $\lim_{t\rightarrow\infty
}u(x,t)=0$, which means that the Markov process dies at infinity, cf. Theorem
\ref{Theorem4}. If $\tau<\infty$, at the time $t=\tau\left(  u^{\left(
0\right)  }\right)  $ the network reaches an absorbing state. We call the
terms $e^{\left(  \widehat{w}_{a}\left(  p^{1-r}\right)  -p^{-1}\overline{\mu
}_{a}\right)  t}\operatorname{Re}\left(  C_{rmj}^{\left(  a\right)  }\left(
0\right)  \Psi_{rmj}\left(  x\right)  \right)  $ \textit{the} \textit{fast
transition modes} because they give rise to fast transitions that drive the
network to an absorbing state (an open compact subset) in a finite time, cf.
Theorem \ref{Theorem2}. The transitions modes are produced by the ultrametric
structure of the basins.

On the other hand, each term $e^{\left(  \widehat{w}_{a}\left(  p^{1-r}%
\right)  -p^{-1}\overline{\mu}_{a}\right)  t}\operatorname{Re}\left(
C_{rmj}^{\left(  a\right)  }\left(  0\right)  \Psi_{rmj}\left(  x\right)
\right)  $ corresponds to the response of a second order type system, which
has a time constant%
\[
\sigma_{a,r}=\frac{4}{-\left(  \widehat{w}_{a}\left(  p^{1-r}\right)
-p^{-1}\overline{\mu}_{a}\right)  }\text{. }%
\]
We argue that the dynamics of each basin is determined by the interaction of a
countable number of second order type systems, each of them with particular
time constant $\sigma_{a,r}$, $a\in\mathcal{G}$, $r\leq1$. The dynamics
between basins is controlled by a Markov chain determined by the matrix
$\Lambda$. The network has a countable number of time scales: $\left\{
\sigma_{a,r}\right\}  $, and $\frac{-4}{\mu}$, where $\mu$ runs through the
eigenvalues of $\Lambda$.

As an application, in Section \ref{Section_7}, we present a toy model \ of
protein folding using a binary network with states $U$ (unfolded state), $N$
(native state). We give an initial concentration $v^{0}(x)$ so that the
average concentration at the time zero in the basin $U$ is non-zero while the
average concentration, at time zero, in the basin $N$ is approximated zero. We
determine explicitly the absorbing \ time $\tau<\infty$ such that the network
gets trapped in an open compact subset of the basin $N$.

Simulation techniques based on Markov state models and transition network
models are widely used to analyze systems evolving in time, for instance, in
molecular dynamics, and particularly in protein folding, see \cite{Bowman et
al}, \cite{Hongyou et al}, \cite{Husic et al}, \cite{Noe et al}, \cite{Prinz
et al}, and the references therein. Our results fit perfectly in this
framework. We start with `directed kinetic transition network model'\ based on
a master equation as in \cite{Zhou et al}, but we construct a\ solvable model
instead of a simulation model. According to Bowman and Pande \cite{Bowman and
Pande}, protein native states are kinetic hubs that can be reached quickly
from any other state. The existence of fast transition modes implies that
certain states on an ultrametric network work as kinetic hubs. Here we mention
again that the existence of kinetic hubs is a consequence of the hierarchical
structure of the energy landscape, and that this type of conclusion cannot be
obtained, as far as we know, by using standard techniques of Markov state
models and transition network models. Furthermore, a native hub allows rapid
folding but proteins can still get stuck in a web of nonnative states,
\cite{Bowman and Pande}. The long-term behavior of an ultrametric network
allows explaining this phenomenon.

Avetisov, Kozyrev et al and Bikulov, Zubarev developed $p$-adic models of
relaxation of complex systems and applied them to protein dynamics
\cite{Av-4}-\cite{Av-5}, \cite{Bikulov et al}, see also \cite{KKZuniga},
\cite{Kozyrev SV} and the references therein. The final purpose of these works
is to provide a $p$-adic analytical description of experiments on the kinetics
of CO binding to myoglobin, which were carried out by the group of
Frauenfelder more than 30 years ago. The $p$-adic description of protein
dynamics fits also in the experiments of J. Friedrich's group on the spectral
diffusion in proteins. Therefore, ultrametric diffusion fits the protein
dynamics well over an extremely large range of time scales from microseconds
to weeks, see e.g. \cite{Avetisov et al} and the references therein. In this
article we consider the protein dynamics from a different perspective,
furthermore, the reaction-diffusion equations used in the above mentioned
works involve the Vladimirov operator, while the reaction-diffusion equations
used here involve a transition kernels, see \cite{To-Zuniga-1}.

The study of $p$-adic heat equations and the attached Markov processes is a
relevant mathematical matter \cite{Bendikov}, \cite{Bradley},
\cite{Dra-Kh-K-V}, \cite{Kozyrev SV}, \cite{To-Zuniga-1}, \cite{To-Zuniga-2},
\cite{V-V-Z}, \cite{Zuniga-LNM-2016}, \cite{Zuniga-Nonlinear},
\cite{Zuniga-Galindo-JMAA}, and the references therein. Recently new models in
geophysics and epidemiology has been formulated using $p$-adic diffusion and
$p$-adic master equations, see \cite{Khrennikov et al 1}-\cite{Khrennikov et
al 3}. The $p$-adic model of ultrametric diffusion with sink functions was
earlier introduced and studied in the context of Eigen's paradox in prebiotic
evolution, which is an analog of the Levinthal paradox in protein folding
\cite{Av-Zhu}. For an in-depth discussion of Eigen's paradox the reader may
consult \cite{Zuniga-Galindo-PNAS} and the references therein.

The study of $p$-adic heat equations on $p$-adic manifolds is an open research
area. The heat equations considered here are heat equations on $p$-adic
compact manifolds. Indeed, a $p$-adic compact manifold is a disjoint union of
a finite numbers of open compact balls, see e.g. \cite[Lemma 7.5.1]{Igusa}.

\section{\label{Section_1}Basic facts on $p$-adic analysis}

In this section we collect some basic results about $p$-adic analysis that
will be used in the article. For an in-depth review of the $p$-adic analysis
the reader may consult \cite{Alberio et al}, \cite{Taibleson}, \cite{V-V-Z}.

\subsection{The field of $p$-adic numbers}

Along this article $p$ will denote a prime number. The field of $p-$adic
numbers $%
\mathbb{Q}
_{p}$ is defined as the completion of the field of rational numbers
$\mathbb{Q}$ with respect to the $p-$adic norm $|\cdot|_{p}$, which is defined
as
\[
\left\vert x\right\vert _{p}=\left\{
\begin{array}
[c]{lll}%
0 & \text{if} & x=0\\
&  & \\
p^{-\gamma} & \text{if} & x=p^{\gamma}\frac{a}{b}\text{,}%
\end{array}
\right.
\]
where $a$ and $b$ are integers coprime with $p$. The integer $\gamma:=ord(x)
$, with $ord(0):=+\infty$, is called the\textit{\ }$p-$\textit{adic order of}
$x$.

Any $p-$adic number $x\neq0$ has a unique expansion of the form
\[
x=p^{ord(x)}\sum_{j=0}^{\infty}x_{j}p^{j},
\]
where $x_{j}\in\{0,\dots,p-1\}$ and $x_{0}\neq0$. By using this expansion, we
define \textit{the fractional part of }$x\in\mathbb{Q}_{p}$, denoted as
$\{x\}_{p}$, to be the rational number
\[
\left\{  x\right\}  _{p}=\left\{
\begin{array}
[c]{lll}%
0 & \text{if} & x=0\text{ or }ord(x)\geq0\\
&  & \\
p^{ord(x)}\sum_{j=0}^{-ord_{p}(x)-1}x_{j}p^{j} & \text{if} & ord(x)<0.
\end{array}
\right.
\]
In addition, any non-zero $p-$adic number can be represented uniquely as
$x=p^{ord(x)}ac\left(  x\right)  $ where $ac\left(  x\right)  =\sum
_{j=0}^{\infty}x_{j}p^{j}$, $x_{0}\neq0$, is called the \textit{angular
component} of $x$. Notice that $\left\vert ac\left(  x\right)  \right\vert
_{p}=1$.

For $r\in\mathbb{Z}$, denote by $B_{r}(a)=\{x\in%
\mathbb{Q}
_{p};\left\vert x-a\right\vert _{p}\leq p^{r}\}$ \textit{the ball of radius
}$p^{r}$ \textit{with center at} $a\in%
\mathbb{Q}
_{p}$, and take $B_{r}(0):=B_{r}$. The ball $B_{0}$ equals $\mathbb{Z}_{p}$,
\textit{the ring of }$p-$\textit{adic integers of }$%
\mathbb{Q}
_{p}$. We also denote by $S_{r}(a)=\{x\in\mathbb{Q}_{p};|x-a|_{p}=p^{r}\}$
\textit{the sphere of radius }$p^{r}$ \textit{with center at} $a\in%
\mathbb{Q}
_{p}$, and take $S_{r}(0):=S_{r}$. We notice that $S_{0}=\mathbb{Z}%
_{p}^{\times}$ (the group of units of $\mathbb{Z}_{p}$). The balls and spheres
are both open and closed subsets in $%
\mathbb{Q}
_{p}$. In addition, two balls in $%
\mathbb{Q}
_{p}$ are either disjoint or one is contained in the other.

The metric space $\left(
\mathbb{Q}
_{p},\left\vert \cdot\right\vert _{p}\right)  $ is a complete ultrametric
space. As a topological space $\left(
\mathbb{Q}
_{p},|\cdot|_{p}\right)  $ is totally disconnected, i.e. the only connected
subsets of $%
\mathbb{Q}
_{p}$ are the empty set and the points. In addition, $\mathbb{Q}_{p}$\ is
homeomorphic to a Cantor-like subset of the real line, see e.g. \cite{Alberio
et al}, \cite{V-V-Z}. A subset of $\mathbb{Q}_{p}$ is compact if and only if
it is closed and bounded in $\mathbb{Q}_{p}$, see e.g. \cite[Section
1.3]{V-V-Z}, or \cite[Section 1.8]{Alberio et al}. The balls and spheres are
compact subsets. Thus $\left(
\mathbb{Q}
_{p},|\cdot|_{p}\right)  $ is a locally compact topological space.

\begin{notation}
We will use $\Omega\left(  p^{-r}|x-a|_{p}\right)  $ to denote the
characteristic function of the ball $B_{r}(a)=a+p^{-r}\mathbb{Z}_{p}$.
\end{notation}

\subsection{Some function spaces}

A complex-valued function $\varphi$ defined on $%
\mathbb{Q}
_{p}$ is \textit{called locally constant} if for any $x\in%
\mathbb{Q}
_{p}$ there exist an integer $l(x)\in\mathbb{Z}$ such that
\begin{equation}
\varphi(x+x^{\prime})=\varphi(x)\text{ for any }x^{\prime}\in B_{l(x)}.
\label{local_constancy_parameter}%
\end{equation}
\ A function $\varphi:%
\mathbb{Q}
_{p}\rightarrow\mathbb{C}$ is called a \textit{Bruhat-Schwartz function (or a
test function)} if it is locally constant with compact support. In this case,
we can take $l=l(\varphi)$ in (\ref{local_constancy_parameter}) independent of
$x$, the largest of such integers is called \textit{the parameter of local
constancy} of $\varphi$. The $\mathbb{C}$-vector space of Bruhat-Schwartz
functions is denoted by $\mathcal{D}(%
\mathbb{Q}
_{p})$. We will denote by $\mathcal{D}_{\mathbb{R}}(%
\mathbb{Q}
_{p})$, the $\mathbb{R}$-vector space of test functions.

Since $(\mathbb{Q}_{p},+)$ is a locally compact topological group, there
exists a Borel measure $dx$, called the Haar measure of $(\mathbb{Q}_{p},+)$,
unique up to multiplication by a positive constant, such that $\int_{U}dx>0 $
for every non-empty Borel open set $U\subset\mathbb{Q}_{p}$, and satisfying
$\int_{E+z}dx=\int_{E}dx$ for every Borel set $E\subset\mathbb{Q}_{p}$, see
e.g. \cite[Chapter XI]{Halmos}. If we normalize this measure by the condition
$\int_{\mathbb{Z}_{p}}dx=1$, then $dx$ is unique. From now on we denote by
$dx$ the normalized Haar measure of $(\mathbb{Q}_{p},+)$.

Given an open subset $U\subset%
\mathbb{Q}
_{p}$, we denote by $L^{\rho}\left(  U\right)  $, with $\rho\in\left[
1,\infty\right)  $, the $%
\mathbb{C}
-$vector space of all the complex valued functions $g$ defined on
$U$\ satisfying
\[
\left\Vert g\right\Vert _{\rho}=\left\{
{\displaystyle\int\limits_{U}}
\left\vert g\left(  x\right)  \right\vert ^{\rho}dx\right\}  ^{\frac{1}{\rho}%
}<\infty.
\]
The corresponding $\mathbb{R}$-vector spaces are denoted as $L_{\mathbb{R}%
}^{\rho}\left(  U\right)  $.

Let $U$ be an open subset of $%
\mathbb{Q}
_{p}$, we denote by $\mathcal{D}(U)$ the $\mathbb{C}$-vector space of all test
functions with support in $U$. Then $\mathcal{D}(U)$ is dense in $L^{\rho
}\left(  U\right)  $, with $\rho\in\left[  1,\infty\right)  $, see e.g.
\cite[Proposition 4.3.3]{Alberio et al}.

\subsection{Fourier transform}

Set $\chi_{p}(y)=\exp(2\pi i\{y\}_{p})$ for $y\in%
\mathbb{Q}
_{p}$. The map $\chi_{p}$ is an additive character on $%
\mathbb{Q}
_{p}$, i.e. a continuous map from $\left(
\mathbb{Q}
_{p},+\right)  $ into $S$ (the unit circle considered as multiplicative group)
satisfying $\chi_{p}(x_{0}+x_{1})=\chi_{p}(x_{0})\chi_{p}(x_{1})$,
$x_{0},x_{1}\in%
\mathbb{Q}
_{p}$. The additive characters of $%
\mathbb{Q}
_{p}$ form an Abelian group which is isomorphic to $\left(
\mathbb{Q}
_{p},+\right)  $, the isomorphism is given by $\xi\rightarrow\chi_{p}(\xi x)$,
see e.g. \cite[Section 2.3]{Alberio et al}.

If $f\in L^{1}\left(
\mathbb{Q}
_{p}\right)  $ its Fourier transform is defined by
\[
(\mathcal{F}f)(\xi)=%
{\displaystyle\int\limits_{\mathbb{Q} _{p}}}
\chi_{p}(\xi x)f(x)dx,\quad\text{for }\xi\in%
\mathbb{Q}
_{p}.
\]
We will also use the notation $\mathcal{F}_{x\rightarrow\xi}f$ and
$\widehat{f}$\ for the Fourier transform of $f$. The Fourier transform is a
linear isomorphism (algebraic and topological) from $\mathcal{D}(%
\mathbb{Q}
_{p})$ onto itself satisfying
\begin{equation}
(\mathcal{F}(\mathcal{F}f))(\xi)=f(-\xi), \label{FF(f)}%
\end{equation}
for every $f\in\mathcal{D}(%
\mathbb{Q}
_{p}),$ see e.g. \cite[Section 4.8]{Alberio et al}. If $f\in L^{2},$ its
Fourier transform is defined as
\[
(\mathcal{F}f)(\xi)=\lim_{k\rightarrow\infty}%
{\displaystyle\int\limits_{|x|_{p}\leq p^{k}}}
\chi_{p}(\xi x)f(x)dx,\quad\text{for }\xi\in%
\mathbb{Q}
_{p},
\]
where the limit is taken in $L^{2}(%
\mathbb{Q}
_{p})$. We recall that the Fourier transform is unitary on $L^{2}(%
\mathbb{Q}
_{p}),$ i.e. $||f||_{2}=||\mathcal{F}f||_{2}$ for $f\in L^{2}(%
\mathbb{Q}
_{p})$ and that (\ref{FF(f)}) is also valid in $L^{2}(%
\mathbb{Q}
_{p})$, see e.g. \cite[Chapter $III$, Section 2]{Taibleson}.

\section{\label{Section_2}A class of $p$-adic heat equations}

Let $\mathcal{K}\subset\mathbb{Z}_{p}$ be a compact open subset. Then,
$\mathcal{K}$ is a finite union of balls contained in $\mathbb{Z}_{p}$. We
assume that%
\[
\mathcal{K}=\bigsqcup\limits_{a\in\mathcal{G}}\left(  a+p\mathbb{Z}%
_{p}\right)  ,
\]
where $\mathcal{G}\subset\left\{  0,1,\ldots,p-1\right\}  $ is a fixed set.
All our results can be easily extended to the general case $%
{\displaystyle\bigsqcup\nolimits_{a\in\mathcal{G}}}
\left(  a+p^{l_{a}}\mathbb{Z}_{p}\right)  $,\ but taking $l_{a}=1$ allows us
to get simpler formulas without sacrificing generality.

We set
\[
\mathcal{C}(\mathcal{K},\mathbb{R}):=\left\{  f:\mathcal{K}\rightarrow
\mathbb{R};f\text{ continuous}\right\}  .
\]
Then $\left(  \mathcal{C}(\mathcal{K},\mathbb{R}),\left\Vert \cdot\right\Vert
_{\infty}\right)  $ endowed with the norm $\left\Vert f\right\Vert _{\infty
}=\max_{x\in\mathcal{K}}\left\vert f\left(  x\right)  \right\vert $ is a
Banach space. Notice that
\begin{equation}
\mathcal{C}(\mathcal{K},\mathbb{R})=\bigoplus\limits_{a\in\mathcal{G}%
}\mathcal{C}\left(  a+p\mathbb{Z}_{p},\mathbb{R}\right)  ,\label{Eq_1A}%
\end{equation}
and similarly
\begin{equation}
\mathcal{C}(\mathcal{K\times K},\mathbb{R})=\bigoplus\limits_{a\in\mathcal{G}%
}\bigoplus\limits_{b\in\mathcal{G}}\mathcal{C}\left(  \left(  a+p\mathbb{Z}%
_{p}\right)  \times\left(  b+p\mathbb{Z}_{p}\right)  ,\mathbb{R}\right)
.\label{Eq_1B}%
\end{equation}
We set $\mathbb{R}_{+}:=\left\{  t\in\mathbb{R};t\geq0\right\}  $. We now pick
two functions $j(x\mid y)$, $j\left(  y\mid x\right)  \in\mathcal{C}%
(\mathcal{K\times K},\mathbb{R}_{+})$. The function $j(y\mid x)$ gives the
transition density rate (per unit of time) from $x$ to $y$, i.e.
\[
\mathbb{P}(x,B)=%
{\textstyle\int\limits_{B}}
j(y\mid x)dy
\]
is the transition probability from $x$ into $B$ (per unit of time), where $B$
is a Borel subset of $\mathcal{K}$. In general the functions $j(x\mid y)$,
$j\left(  y\mid x\right)  $ are different.

We now define the operator%
\begin{equation}
\boldsymbol{J}f(x)=\int\limits_{\mathcal{K}}\left\{  j(x\mid y)f(y)-j(y\mid
x)f(x)\right\}  dy\text{, }f\in\mathcal{C}(\mathcal{K},\mathbb{R})\text{.}
\label{Operator_J}%
\end{equation}

\begin{lemma}
The mapping $f\rightarrow\boldsymbol{J}f$ gives rise to a well-defined,
linear, bounded operator from $\mathcal{C}(\mathcal{K},\mathbb{R})$\ into itself.
\end{lemma}

Indeed, by using that
\begin{align*}
\left\vert j(x\mid y)f(y)-j(y\mid x)f(x)\right\vert  &  \leq\left(
\max_{x,y\in\mathcal{K}}j(x\mid y)+\max_{x,y\in\mathcal{K}}j(y\mid x)\right)
\left(  \max_{x\in\mathcal{K}}\left\vert f(x)\right\vert \right)
1_{\mathcal{K}}\left(  y\right) \\
&  =A\left\Vert f\right\Vert _{\infty}1_{\mathcal{K}}\left(  y\right)  ,
\end{align*}
it follows from the dominated convergence theorem that $\boldsymbol{J}$ is a
bounded operator from $\mathcal{C}(\mathcal{K},\mathbb{R})$ into itself.

We now assume that
\begin{equation}
j(x\mid y)\leq j(y\mid x)\text{ for any }x,y\in\mathcal{K}\text{.}
\tag{Hypothesis 1}%
\end{equation}
Under the Hypothesis 1, the operator $\boldsymbol{J}$ satisfies the positive
maximum principle, i.e. \ if $h\in\mathcal{C}(\mathcal{K},\mathbb{R})$ and
$\max_{x\in\mathcal{K}}h\left(  x\right)  =h(x_{0})\geq0$, then $\left(
\boldsymbol{J}h\right)  (x_{0})\leq0$. Now, for any fixed $\lambda>\left\Vert
\boldsymbol{J}\right\Vert $, the operator $\frac{\boldsymbol{1}}%
{\boldsymbol{1}-\frac{1}{\lambda}\boldsymbol{J}}=1+\lambda\boldsymbol{J+}%
\lambda^{2}\boldsymbol{J}^{2}+\cdots+\lambda^{n}\boldsymbol{J}^{n}+\cdots$ is
linear and bounded, consequently $\mathrm{rank}\left(  \boldsymbol{1}-\frac
{1}{\lambda}\boldsymbol{J}\right)  =\mathcal{C}(\mathcal{K},\mathbb{R})$. Now,
by using the Hille-Yoshida-Ray theorem, see e.g. \cite[Theorem 2.18]{Taira},
$\boldsymbol{J}$ generates a strongly continuous, positive, contraction
semigroup $\left\{  e^{t\boldsymbol{J}}\right\}  _{t\geq0}$ on $\mathcal{C}%
(\mathcal{K},\mathbb{R})$, which is a Feller semigroup. For the sake of
completeness we recall de definition of Feller semigroup here.

\begin{definition}
A family of bounded linear operators $\boldsymbol{P}_{t}:\mathcal{C}\left(
\mathcal{K},\mathbb{R}\right)  \rightarrow\mathcal{C}\left(  \mathcal{K}%
,\mathbb{R}\right)  $, $t\geq0$, is called a Feller semigroup if

\noindent(i) $\boldsymbol{P}_{s+t}=\boldsymbol{P}_{s}\boldsymbol{P}_{t}$ and
$\boldsymbol{P}_{0}=I$;

\noindent(ii) $\lim_{t\rightarrow0}||\boldsymbol{P}_{t}h-h||_{\infty}=0$ for
any $h\in\mathcal{C}\left(  \mathcal{K},\mathbb{R}\right)  $;

\noindent(iii) $0\leq\boldsymbol{P}_{t}h\leq1$ if $0\leq h\leq1$, with
$h\in\mathcal{C}\left(  \mathcal{K},\mathbb{R}\right)  $ and for any $t\geq0$.
\end{definition}

The Hille-Yoshida-Ray theorem characterizes the Feller semigroups. Now by
using the correspondence between Feller semigroups and transition functions,
there exists a uniformly stochastically continuous $C_{0}$-transition function
$p_{t}\left(  x,dy\right)  $ satisfying condition $(L)$, such that
\[
e^{t\boldsymbol{J}}u_{0}\left(  x\right)  =\int\limits_{\mathcal{K}}%
u_{0}(y)p_{t}\left(  x,dy\right)  \text{ for }u_{0}\in\mathcal{C}\left(
\mathcal{K}_{N},\mathbb{R}\right)  \text{,}
\]
see e.g. \cite[Theorem 2.15]{Taira}. Now, by using the correspondence between
transition functions and Markov processes, there exists a strong Markov
process $\mathfrak{X}$ whose paths are right continuous and have no
discontinuities other than jumps, see e.g. \cite[Theorem 2.12]{Taira}. Then we
have established the following result:

\begin{theorem}
\label{Theorem1} Assuming Hypothesis 1, with $T\in\left[  0,\infty\right]  $.
Consider the Cauchy problem:$.$%
\begin{equation}
\left\{
\begin{array}
[c]{ll}%
u\left(  \cdot,t\right)  \in\mathcal{C}^{1}\left(  \left[  0,T\right]
,\mathcal{C}\left(  \mathcal{K},\mathbb{R}\right)  \right)  ; & \\
& \\
\frac{du\left(  x,t\right)  }{dt}=\int\limits_{\mathcal{K}}\left\{  j(x\mid
y)u(y,t)-j(y\mid x)u(x,t)\right\}  dy, & t\in\left[  0,T\right]
,x\in\mathcal{K};\\
& \\
u\left(  x,0\right)  =u_{0}\left(  x\right)  \in\mathcal{C}\left(
\mathcal{K},\mathbb{R}_{+}\right)  . &
\end{array}
\right.  \label{Eq_Cauchy_problem}%
\end{equation}
There exists a probability measure $p_{t}\left(  x,\cdot\right)  $,
$t\in\left[  0,T\right]  $, with $T=T(u_{0})$, $x\in\mathcal{K}$, on the Borel
$\sigma$-algebra of $\mathcal{K}$, such that the Cauchy problem
(\ref{Eq_Cauchy_problem}) has a unique solution of the form%
\[
h(x,t)=\int\limits_{\mathcal{K}}u_{0}(y)p_{t}\left(  x,dy\right)  .
\]
In addition, $p_{t}\left(  x,\cdot\right)  $ is the transition function of a
Markov process $\mathfrak{X}$ whose paths are right continuous and have no
discontinuities other than jumps.
\end{theorem}

Assuming that $u_{0}(y)$ is the characteristic function of a Borel set $B$,
and that at time zero the system was in a state belonging to $B$, then,
$h(x,t)$ is the probability that the system be in a state belonging to $B$ at
the time $t$, with $t\leq T$.

\subsection{Further remarks}

The points of $\mathcal{K}$ are the states of a Markovian system which moves
randomly between these states. Let $\mathcal{B}(\mathcal{K})$ denote the Borel
$\sigma$-algebra of $\mathcal{K}$. Then, for $x\in\mathcal{K}$, $E\in
\mathcal{B}(\mathcal{K})$, $p_{t}\left(  x,E\right)  $ gives the transition
probability that the system starting at state $x$ will move to a state
belonging to set $E$ at time $t$. A Markov transition function $p_{t}\left(
x,\cdot\right)  $ is said to be \textit{conservative} if it satisfies%
\[
p_{t}\left(  x,\mathcal{K}\right)  =1\text{ for }t\geq0\text{ and each }%
x\in\mathcal{K}\text{.}%
\]
The transition function given in Theorem \ref{Theorem1} is not necessarily
conservative, i.e. $p_{t}\left(  x,\mathcal{K}\right)  \leq1$ for $t\geq0$ and
each $x\in\mathcal{K}$. In this case, the stochastic process $\mathfrak{X}$
may die, which means that the system may get trapped in a terminal state
$\Delta\notin\mathcal{K}$. We take $\mathcal{K}_{\Delta}:=\mathcal{K}%
{\textstyle\bigsqcup}
\left\{  \Delta\right\}  $, and extend $p_{t}\left(  x,\cdot\right)  $ to a
Markov transition function $\widetilde{p}_{t}\left(  x,\cdot\right)  $ on
$\mathcal{K}_{\Delta}$\ as follows:%
\[
\left\{
\begin{array}
[c]{ll}%
\widetilde{p}_{t}\left(  x,E\right)  =p_{t}\left(  x,E\right)  \text{,} &
x\in\mathcal{K}\text{, }E\in\mathcal{B}(\mathcal{K})\text{;}\\
& \\
\widetilde{p}_{t}\left(  x,\left\{  \Delta\right\}  \right)  =1-p_{t}\left(
x,\mathcal{K}\right)  \text{,} & x\in\mathcal{K}\text{;}\\
& \\
\widetilde{p}_{t}\left(  \partial,\mathcal{K}\right)  =0; & \\
& \\
\widetilde{p}_{t}\left(  \Delta,\left\{  \Delta\right\}  \right)  =1, &
\end{array}
\right.
\]
for any $t\geq0$.

Intuitively, the system assumes configurations in space $\mathcal{K}$ until it
gets trapped in the configuration $\Delta$. The random variable%
\[
\varkappa\left(  \omega\right)  :=\inf\left\{  t\in\left[  0,\infty\right]
;\mathfrak{X}_{t}(\omega)=\Delta\right\}
\]
is called the \textit{lifetime} of the process $\mathfrak{X}$. Notice that
$T\leq\varkappa\left(  \omega\right)  $ in Theorem \ref{Theorem1}, which means
that the integro-differential equation in (\ref{Eq_Cauchy_problem}) does not
have physical meaning for $t>\varkappa\left(  \omega\right)  $. The state
$\Delta$ is a particular case of an absorbing state, which is a $\Gamma
\in\mathcal{K}_{\Delta}$ such that $\widetilde{p}_{t}\left(  \Gamma,\left\{
\Gamma\right\}  \right)  =1$, for any $t\geq0$, see \cite[Lemma 5.3]{Dynkin}.

Finally, the integro-differential equation in (\ref{Eq_Cauchy_problem}) is a
master equation in an ultrametric space, see \cite[Chaptre V]{van Kampen},
\cite{Av-4}.

\section{\label{Section_3}$p$-Adic heat equations on balls}

Given a function $f\in\mathcal{C}(p\mathbb{Z}_{p},\mathbb{R})$, we extend it
to $\mathbb{Q}_{p}$ by taking $f(x)=0$ for $x\notin p\mathbb{Z}_{p}$. In this
way $f\in L^{1}(p\mathbb{Z}_{p})\cap L^{2}(p\mathbb{Z}_{p})\subset
L^{1}(\mathbb{Q}_{p})\cap L^{2}(\mathbb{Q}_{p})$, and thus the Fourier
transform $\widehat{f}$ is defined in the standard way.

Given a radial function $w\in\mathcal{C}(p\mathbb{Z}_{p},\mathbb{R}_{+})$, we
define the operator
\[
\boldsymbol{W}f\left(  x\right)  =\int\limits_{p\mathbb{Z}_{p}}\left\{
f(y)-f(x)\right\}  w(\left\vert x-y\right\vert _{p})dy\text{ for }%
f\in\mathcal{C}(p\mathbb{Z}_{p},\mathbb{R}).
\]
Since $p\mathbb{Z}_{p}$ is an additive group, $\boldsymbol{W}$ is a linear
bounded operator from $\mathcal{C}(p\mathbb{Z}_{p},\mathbb{R})$ into itself.
We set $\gamma=\int_{p\mathbb{Z}_{p}}$ $w(\left\vert y\right\vert _{p}%
)dy\geq0$, then%
\[
\boldsymbol{W}f\left(  x\right)  =f(x)\ast w(x)-\gamma f(x),
\]
and%
\[
\boldsymbol{W}f\left(  x\right)  =\mathcal{F}_{\xi\rightarrow x}^{-1}(\left(
\widehat{w}(\xi)-\gamma\right)  \mathcal{F}_{x\rightarrow x}f),
\]
i.e. $\boldsymbol{W}$ is a pseudodifferential operator with symbol
$\widehat{w}(\xi)-\gamma$.

The function $\widehat{w}(\xi)=\widehat{w}(\left\vert \xi\right\vert _{p})$ is
radial and integrable, furthermore, it verifies that
\[
\widehat{w}(\xi)-\gamma=\left\{
\begin{array}
[c]{lll}%
0 & \text{if} & \left\vert \xi\right\vert _{p}\leq p\\
&  & \\
-\left(  1-p^{-1}\right)  \sum\limits_{j=1}^{-ord(\xi)-1}p^{-j}w(p^{-j}%
)-p^{ord(\xi)}w\left(  p^{1+ord(\xi)}\right)  & \text{if} & \left\vert
\xi\right\vert _{p}\geq p^{2}.
\end{array}
\right.
\]
Indeed,
\begin{align}
\widehat{w}(\xi)  &  =\int\limits_{p\mathbb{Z}_{p}}\chi_{p}\left(
x\xi\right)  w(\left\vert x\right\vert _{p})dx=\sum\limits_{j=1}^{\infty
}p^{-j}w(p^{-j})\int\limits_{\mathbb{Z}_{p}^{\times}}\chi_{p}\left(  p^{j}%
y\xi\right)  dy\nonumber\\
&  =\sum\limits_{j=1}^{\infty}p^{-j}w(p^{-j})\left\{
\begin{array}
[c]{lll}%
1-p^{-1} & \text{if} & -ord(\xi)\leq j\\
-p^{-1} & \text{if} & -ord(\xi)=j+1\\
0 & \text{if} & -ord(\xi)\geq j+2
\end{array}
\right. \nonumber\\
&  =\left\{
\begin{array}
[c]{lll}%
\left(  1-p^{-1}\right)  \sum\limits_{j=1}^{\infty}p^{-j}w(p^{-j}) & \text{if}
& \left\vert \xi\right\vert _{p}\leq p\\
&  & \\
\left(  1-p^{-1}\right)  \sum\limits_{j=-ord(\xi)}^{\infty}p^{-j}%
w(p^{-j})-p^{ord(\xi)}w\left(  p^{1+ord(\xi)}\right)  & \text{if} & \left\vert
\xi\right\vert _{p}\geq p^{2}.
\end{array}
\right.  \label{Formula_1}%
\end{align}
The announced formula follows from (\ref{Formula_1}).

\subsection{$p$-Adic wavelets}

We now recall that the set of functions $\left\{  \Psi_{rmj}\right\}  $
defined as%
\begin{equation}
\Psi_{rmj}\left(  x\right)  =p^{\frac{-r}{2}}\chi_{p}\left(  p^{-1}j\left(
p^{r}x-m\right)  \right)  \Omega\left(  \left\vert p^{r}x-m\right\vert
_{p}\right)  , \label{Eq_40}%
\end{equation}
where $r\in\mathbb{Z}$, $j\in\left\{  1,\cdots,p-1\right\}  $, and $m$ runs
through a fixed set of representatives of $\mathbb{Q}_{p}/\mathbb{Z}_{p}$, is
an orthonormal basis of $L^{2}(\mathbb{Q}_{p})$. Furthermore,
\begin{equation}%
{\textstyle\int\limits_{\mathbb{Q}_{p}}}
\Psi_{rmj}\left(  x\right)  dx=0, \label{Eq_41}%
\end{equation}
see e.g. \cite[Theorem 3.29]{KKZuniga} or \cite[Theorem 9.4.2]{Alberio et al}.

By using that $\Omega\left(  \left\vert p^{r}x-m\right\vert _{p}\right)
=1\Leftrightarrow x\in p^{-r}m+p^{-r}\mathbb{Z}_{p}$, one gets that the
wavelets supported in $p\mathbb{Z}_{p}$ are exactly those satisfying%
\begin{equation}
r\leq-1\text{ and }m\in p^{r+1}\mathbb{Z}_{p}\cap\mathbb{Q}_{p}/\mathbb{Z}%
_{p}. \label{Condition}%
\end{equation}
The wavelets satisfying $p\mathbb{Z}_{p}\subsetneqq\mathrm{supp}$ $\Psi_{rmj}
$, i.e.%
\begin{equation}
p\mathbb{Z}_{p}\subsetneqq p^{-r}m+p^{-r}\mathbb{Z}_{p}, \label{Condition_3}%
\end{equation}
are those with $m=0$, $r<-1$, $j\in\left\{  1,\ldots,p-1\right\}  $. Indeed,
from (\ref{Condition_3}), $0\in p^{-r}m+p^{-r}\mathbb{Z}_{p}$, and then
$p^{-r}m+p^{-r}\mathbb{Z}_{p}=p^{-r}\mathbb{Z}_{p}$, and $p^{-r}m\in
p^{-r}\mathbb{Z}_{p}$, i.e. $m=0$. Finally, $p\mathbb{Z}_{p}\subsetneqq
p^{-r}\mathbb{Z}_{p}$ implies that $r<-1$. Now the the restriction of these
wavelets to $p\mathbb{Z}_{p}$ is exactly%
\[
p^{\frac{-r}{2}}\Omega\left(  p\left\vert x\right\vert _{p}\right)  .
\]
Finally, under condition (\ref{Condition_3}), with $m\neq0$, the restriction
of $\Psi_{rmj}$\ to $p\mathbb{Z}_{p}$ is the constant function zero.

Any function $f$ from $\mathcal{C}(p\mathbb{Z}_{p},\mathbb{R})$ admits an
expansion of the form%
\begin{equation}
f(x)=\left(  \text{ }\int\limits_{p\mathbb{Z}_{p}}f\left(  x\right)
dx\right)  p^{\frac{1}{2}}\Omega\left(  p\left\vert x\right\vert _{p}\right)
+\sum\limits_{rmj}C_{rmj}\Psi_{rmj}\left(  x\right)  , \label{Expansion}%
\end{equation}
where $r$ and $m$ satisfy (\ref{Condition}) and $j\in\left\{  1,\cdots
,p-1\right\}  $. This result is a variation of a classical result, see e.g.
$\ $\cite[Proposition 1]{Zuniga-Galindo-PNAS} and the references therein.
Furthermore, the wavelets appearing in (\ref{Expansion}) satisfy%
\begin{equation}
\boldsymbol{W}\Psi_{rmj}\left(  x\right)  =\left(  \widehat{w}(p^{1-r}%
)-\gamma\right)  \Psi_{rmj}\left(  x\right)  , \label{Spectra}%
\end{equation}
where $\Psi_{rmj}\left(  x\right)  $\ is a wavelet supported in $p\mathbb{Z}%
_{p}$, see e.g. \cite[Theorem 3.29]{KKZuniga}, \cite[Theorem 9.4.2]{Alberio et
al}. Notice that $\widehat{w}(p^{1-r})-\gamma\leq0$ for $r\leq-1$.

\section{\label{Section_4}$p$-adic Models of Relaxation of Complex Systems}

\subsection{Random walks on complex landscapes}

It is widely accepted that the dynamics of a large class of complex systems
such as glasses and proteins can be described by a random walk on a complex
energy landscape, see e.g. \cite{Fraunfelder et al 2}, \cite{Fraunfelder et al
3}, \cite[and \ the references therein]{Kozyrev SV}, \cite{Wales}.\ An energy
landscape (or simply a landscape) is a continuous function $\mathbb{U}%
:X\rightarrow\mathbb{R}$ that assigns to each physical state of a system its
energy. In many cases we can take $X$ to be a subset of $\mathbb{R}^{N}$. The
term complex landscape means that function $\mathbb{U}$ has many local minima.
In this case the method of \textit{interbasin kinetics} is applied, in this
approach, the study of a random walk on a complex landscape is based on a
description of the kinetics generated by transitions between groups of states
(basins). Minimal basins correspond to local minima of energy, and large
basins have a hierarchical structure. Procedures for constructing hierarchies
of basins and models of interbasin kinetics from an arbitrary energy landscape
have been studied extensively, see e.g. \cite{Becker et al}, \cite{Stillinger
et al 1}, \cite{Stillinger et al 2}. By using these methods, a complex
landscape is approximated by a \textit{disconnectivity graph} (a rooted tree)
and by a function on the tree describing the distribution of the activation
energies. For further details the reader may consult \cite[and \ the
references therein]{Kozyrev SV}, see also \cite{KKZuniga}.

The dynamics of the system is described \ by a master equation on simple
finite graph $\mathcal{T}$ (for instance a disconnectivity graph). Given two
vertices $I$, $J\in V(\mathcal{T})$, $w_{I,J}\geq0$ denotes the transition
rate peer unit of time from state $J$ to state $I$. We set
\[
w_{I}:=\sum\limits_{K\in V(\mathcal{T})}w_{K,I}
\]
and
\[
W=\left[  W_{I,J}\right]  _{I,J\in V(\mathcal{T})}\text{, where }%
W_{I,J}=w_{I,J}-w_{I}\delta_{IJ},
\]
where $\delta_{IJ}$\ denotes the Kronecker delta. The dynamics of the system
is then encoded in a system of kinetic equations (the master equation) of the
form:%
\begin{equation}
\frac{d}{dt}p_{I}\left(  t\right)  =\sum_{J\in V(\mathcal{G})}\left\{
w_{I,J}p_{J}\left(  t\right)  -w_{J,I}p_{I}\left(  t\right)  \right\}  ,\text{
for }I\in V(\mathcal{T}), \label{Equation_Discrete}%
\end{equation}
where $p_{I}\left(  t\right)  $\ is the probability of finding the system at
state $I$ at time $t$. We denote by $p(t)$ the column vector $\left[
p_{I}\left(  t\right)  \right]  _{I\in V(\mathcal{T})}$, then the master
equation takes the form%
\begin{equation}
\frac{d}{dt}p\left(  t\right)  =Wp\left(  t\right)  .
\label{Master_Equation_0}%
\end{equation}
Notice that $W$ is $Q$-matrix, i.e. $W_{I,J}\geq0$ for $I\neq J$, and
$\sum_{I\in V(\mathcal{T})}W_{I,J}=0$. This implies that%
\[
p(t)=e^{tW}p(0)
\]
is the transition function of a time-homogeneous continuous Markov chain, with
initial probability distribution $p(0)$.

In practical applications the matrices $W$ are very large, then the
determination of their spectra is a highly non-trivial problem.

\subsection{$p$-adic continuous versions}

In order to construct a $p$-adic continuous version of master equation
(\ref{Master_Equation_0}), it is sufficient to find a linear bounded operator
$\boldsymbol{J}:\mathcal{C}\left(  \mathcal{K}_{N},\mathbb{R}\right)
\rightarrow\mathcal{C}\left(  \mathcal{K}_{N},\mathbb{R}\right)  $, where
$\mathcal{K}_{N}$ is an open compact subset, such that its restriction
$\boldsymbol{J}_{N}=\boldsymbol{J}\mid_{X_{N}}$ to a finite dimensional vector
space $X_{N}\subset\mathcal{C}\left(  \mathcal{K}_{N},\mathbb{R}\right)  $ be
represented by the matrix $W$ in (\ref{Master_Equation_0}).

We set%
\[
G_{N}:=\mathbb{Z}_{p}/p^{N}\mathbb{Z}_{p}\text{ for }N\geq2\text{.}
\]
We identify $G_{N}$\ with the set of representatives of the form%
\begin{equation}
I=I_{0}+I_{1}p+\ldots+I_{N-1}p^{N-1}, \label{Eq_I}%
\end{equation}
where the $I_{j}$s are $p$-adic digits. We now identify each vertex of
$\mathcal{T}$ with a $p$-adic number of the form (\ref{Eq_I}). From now on, we
take $V(\mathcal{T})=G_{N}^{0}\subset G_{N}$. Here $N$ is a fixed positive
integer such that $\#G_{N}^{0}\leq p^{N}$.

We denote by $\Omega\left(  p^{N}\left\vert x-I\right\vert _{p}\right)  $ the
characteristic function of the ball centered at $I$ with radius $p^{-N}$,
which corresponds to the set $I+p^{N}\mathbb{Z}_{p}$. Now, we attach to
$\mathcal{T}$ the open compact subset \
\[
\mathcal{K}_{N}=%
{\textstyle\bigsqcup\limits_{I\in G_{N}^{0}}}
I+p^{N}\mathbb{Z}_{p}\text{,}
\]
\ and the $\mathbb{R}$-vector space $X_{N}$ generated by the functions
$\left\{  \Omega\left(  p^{N}\left\vert x-I\right\vert _{p}\right)  \right\}
_{I\in G_{N}^{0}}$. Then ${\large X}_{N}$ is the $\mathbb{R}$-vector space
consisting of all the test functions supported in $\mathcal{K}_{N}$\ having
the form%
\[
\varphi\left(  x\right)  =%
{\textstyle\sum\limits_{J\in G_{N}^{0}}}
\varphi\left(  J\right)  \Omega\left(  p^{N}\left\vert x-J\right\vert
_{p}\right)  \text{,}
\]
where $\varphi\left(  J\right)  \in\mathbb{R}$. This is the space of
continuous functions on $\mathcal{T}$. Notice that $X_{N}\subset
\mathcal{C}\left(  \mathcal{K}_{N},\mathbb{R}\right)  $.

We now define%
\begin{align*}
j_{N}(x  &  \mid y)=p^{N}%
{\textstyle\sum\limits_{\substack{I,K\in G_{N}^{0}\\I\neq K}}}
w_{I,K}\Omega\left(  p^{N}\left\vert x-I\right\vert _{p}\right)  \Omega\left(
p^{N}\left\vert y-K\right\vert _{p}\right)  \text{,}\\
j_{N}(y  &  \mid x)=p^{N}%
{\textstyle\sum\limits_{\substack{I,K\in G_{N}^{0}\\I\neq K}}}
w_{I,K}\Omega\left(  p^{N}\left\vert x-K\right\vert _{p}\right)  \Omega\left(
p^{N}\left\vert y-I\right\vert _{p}\right)  ,
\end{align*}
where $x$, $y\in\mathbb{Q}_{p}$, and $W=$ $\left[  w_{I,J}\right]  _{I,J\in
G_{N}^{0}}$. Notice that $j_{N}(x\mid y)$, $j_{N}(y\mid x)$ are test functions
from $\mathcal{D}(\mathcal{K}_{N}\times\mathcal{K}_{N},\mathbb{R})$.

We also define the operator%
\[
\boldsymbol{J}\varphi\left(  x\right)  =%
{\textstyle\int\limits_{\mathcal{K}_{N}}}
\left\{  j_{N}(x\mid y)\varphi\left(  y\right)  -j_{N}(y\mid x)\varphi\left(
x\right)  \right\}  dy\text{, for }\varphi\in\mathcal{D}(\mathcal{K}%
_{N},\mathbb{R})\text{.}%
\]
By using that%
\[%
{\textstyle\int\limits_{\mathcal{K}_{N}}}
j_{N}(x\mid y)\varphi\left(  y\right)  dy=%
{\textstyle\sum\limits_{_{I\in G_{N}^{0}}}}
\left\{
{\textstyle\sum\limits_{\substack{J\in G_{N}^{0}\\J\neq I}}}
w_{I,J}\varphi\left(  J\right)  \right\}  \Omega\left(  p^{N}\left\vert
x-I\right\vert _{p}\right)  \text{,}%
\]%
\[%
{\textstyle\int\limits_{\mathcal{K}_{N}}}
j_{N}(y\mid x)\varphi\left(  x\right)  dy=%
{\textstyle\sum\limits_{_{I\in G_{N}^{0}}}}
\left\{
{\textstyle\sum\limits_{\substack{J\in G_{N}^{0}\\J\neq I}}}
w_{J,I}\varphi\left(  I\right)  \right\}  \Omega\left(  p^{N}\left\vert
x-I\right\vert _{p}\right)  .
\]
one gets that
\[
\boldsymbol{J}_{N}\varphi\left(  x\right)  =%
{\textstyle\sum\limits_{_{I\in G_{N}^{0}}}}
\left\{
{\textstyle\sum\limits_{J\in G_{N}^{0}}}
w_{I,J}\varphi\left(  J\right)  -w_{J,I}\varphi\left(  I\right)  \right\}
\Omega\left(  p^{N}\left\vert x-I\right\vert _{p}\right)  ,
\]
which implies that the operator $\boldsymbol{J}_{N}:X_{N}\rightarrow X_{N}$ is
represented by the matrix $W$, and
\begin{equation}
\frac{d}{dt}u\left(  x,t\right)  =\boldsymbol{J}u\left(  x,t\right)  \text{,
}x\in\mathcal{K}_{N},t\geq0 \label{Master_Equation_0_A}%
\end{equation}
is a continuous version of (\ref{Master_Equation_0}). To construct solvable
models we look for operators $\boldsymbol{J}$ such that
(\ref{Master_Equation_0_A}) be a $p$-adic heat equation, i.e. Theorem
\ref{Theorem1} holds true for (\ref{Master_Equation_0_A}). In the continuous
approach, the matrices $W$ are replaced by the operators $\boldsymbol{J}$. The
determination of the spectra of a large class of these operators can be
obtained using $p$-adic wavelets. This section is based in our work
\cite{Zuniga-Galindo-JMAA}, see also \cite{Nakao-Mikhailov}.

\section{\label{Section_5}$p$-Adic transition networks}

\subsection{Some formulae}

We rewrite the equation (\ref{Eq_Cauchy_problem}) as%

\begin{equation}
\frac{du\left(  x,t\right)  }{dt}=\int\limits_{\mathcal{K}}j(x\mid y)\left\{
u(y,t)-u(x,t)\right\}  dy-u\left(  x,t\right)  S(x), \label{Eq_Master_1}%
\end{equation}
where%
\[
S(x):=\left\{  \int\limits_{\mathcal{K}}j(y\mid x)-j(x\mid y)\right\}
dy\geq0.
\]
By using (\ref{Eq_1A}) and assuming that $u\left(  \cdot,t\right)  $
$\in\mathcal{C}(\mathcal{K},\mathbb{R})$,%
\[
u\left(  x,t\right)  =\sum\limits_{b\in\mathcal{G}}u_{b}(x,t)\text{, where
}u_{b}(\cdot,t)\in\mathcal{C}\left(  b+p\mathbb{Z}_{p},\mathbb{R}\right)  ,
\]
and $u_{b}(x,t)=u\left(  x,t\right)  $ for $x\in b+p\mathbb{Z}_{p}$. Now by
using (\ref{Eq_1B}),
\[
j\left(  x\mid y\right)  =\sum\limits_{a\in\mathcal{G}}\sum\limits_{b\in
\mathcal{G}}j_{a,b}(x\mid y)\text{, where }j_{a,b}(x\mid y)\in\mathcal{C}%
\left(  \left(  a+p\mathbb{Z}_{p}\right)  \times\left(  b+p\mathbb{Z}%
_{p}\right)  ,\mathbb{R}_{+}\right)  ,
\]
and%
\[
j\left(  y\mid x\right)  =\sum\limits_{a\in\mathcal{G}}\sum\limits_{b\in
\mathcal{G}}j_{b,a}(y\mid x)\text{, where }j_{b,a}(y\mid x)\in\mathcal{C}%
\left(  \left(  a+p\mathbb{Z}_{p}\right)  \times\left(  b+p\mathbb{Z}%
_{p}\right)  ,\mathbb{R}_{+}\right)  .
\]
We set
\[
\boldsymbol{J}_{0}f(x):=\int\limits_{\mathcal{K}}j(x\mid y)\left\{
f(y)-f(x)\right\}  dy\text{ for }f\left(  x\right)  =\sum\limits_{b\in
\mathcal{G}}f_{b}(x)\in\mathcal{C}\left(  \mathcal{K},\mathbb{R}\right)  .
\]
The following formula for operator $\boldsymbol{J}_{0}$ holds true:%
\begin{equation}
\boldsymbol{J}_{0}f_{a}(x)=\sum\limits_{b\in\mathcal{G}}\text{ }%
\int\limits_{b+p\mathbb{Z}_{p}}\left\{  f_{b}(y)-f_{a}(x)\right\}
j_{a,b}(x\mid y)dy. \label{Eq_Formula_j_a}%
\end{equation}
Indeed, for $x\in a+p\mathbb{Z}_{p}$, $f(x)=f_{a}(x)$ and%
\begin{gather*}
\boldsymbol{J}_{0}f_{a}(x)=\int\limits_{\mathcal{K}}\left(  f(y)-f_{a}%
(x)\right)  j(x\mid y)dy=\int\limits_{\mathcal{K}}\left(  \sum\limits_{b\in
\mathcal{G}}f_{b}(y)-f_{a}(x)\right)  j(x\mid y)dy\\
=\int\limits_{\mathcal{K}}\sum\limits_{b\in\mathcal{G}}f_{b}(y)j(x\mid
y)dy-\int\limits_{\mathcal{K}}f_{a}(x)j(x\mid y)dy\\
=\int\limits_{\mathcal{K}}\sum\limits_{b\in\mathcal{G}}f_{b}(y)j_{a,b}(x\mid
y)dy-\int\limits_{\mathcal{K}}f_{a}(x)j(x\mid y)dy\\
=\sum\limits_{b\in\mathcal{G}}\text{ }\int\limits_{b+p\mathbb{Z}_{p}}%
f_{b}(y)j_{a,b}(x\mid y)dy-\sum\limits_{b\in\mathcal{G}}\text{ }%
\int\limits_{b+p\mathbb{Z}_{p}}f_{a}(x)j_{a,b}(x\mid y)dy\\
=\sum\limits_{b\in\mathcal{G}}\text{ }\int\limits_{b+p\mathbb{Z}_{p}}\left(
f_{b}(y)-f_{a}(x)\right)  j_{a,b}(x\mid y)dy.
\end{gather*}
By using (\ref{Eq_Master_1})-\ref{Eq_Formula_j_a}), and the fact that
$S\left(  x\right)  =\sum\limits_{b\in\mathcal{G}}S_{b}(x)$, with $S_{b}%
\in\mathcal{C}\left(  b+p\mathbb{Z}_{p},\mathbb{R}\right)  $, one obtains that
(\ref{Eq_Master_1}) is equivalent to%
\begin{equation}
\frac{du_{a}\left(  x,t\right)  }{dt}=\sum\limits_{b\in\mathcal{G}}\left\{
\text{ }\int\limits_{b+p\mathbb{Z}_{p}}\left\{  u_{b}(y,t)-u_{a}(x,t)\right\}
j_{a,b}(x\mid y)dy\right\}  -u_{a}(x,t)S_{a}(x), \label{Eq_Master_2}%
\end{equation}
for $a\in\mathcal{G}$.

\subsection{Definition of $p$-adic transition network}

A $p$\textit{-adic transition network (or an ultrametric network, or simply a
network)} is a model of a complex system consisting of three components: an
energy landscape, a Markov process on the energy landscape, and a master
equation. The energy landscape consists of a finite number of \textit{basins}
$a+p\mathbb{Z}_{p}$, $a\in\mathcal{G}$. Each basin is formed by infinitely
many \textit{network} \textit{configurations} organized hierarchically in a
tree. We use the words \textit{state}, \textit{conformational substate} as
synonyms of \textit{configuration}. The transitions between the basins are
determined by \textit{the transition density matrix of the network} as
$\left[  j_{a,b}(x\mid y)\right]  _{a,b\in\mathcal{G}}$, and the \textit{sink
function} as $\left[  S_{a}(x)\right]  _{a\in\mathcal{G}}$. The function
$j_{a,b}(x\mid y)\geq0$ is the transition density (per unit of time) that the
network goes from configuration $y\in b+p\mathbb{Z}_{p}$ to configuration
$x\in a+p\mathbb{Z}_{p}$.

The Markov process in the energy landscape encodes the temporal evolution of
the network as random transitions between configurations from the energy
landscape. The master equation describe the time evolution of the density of
the configurations. Some authors use the term population instead of density.
We denote by $u_{a}(x,t)$, for $x\in a+p\mathbb{Z}_{p}$, the density of
configurations in the basin $a+p\mathbb{Z}_{p}$ at the time $t$.
The\textit{\ master equation of the network} is given by the system
(\ref{Eq_Master_2}).

We now set $\widetilde{u}_{a}(x,t):=u_{a}(a+x,t)$ \ for $x\in p\mathbb{Z}_{p}%
$,%
\[
\widetilde{j}_{a,b}\left(  x\mid y\right)  :=j_{a,b}\left(  a+x\mid
b+y\right)  \text{, \ \ }\widetilde{j}_{b,a}\left(  y\mid x\right)
:=j_{b,a}\left(  b+y\mid a+x\right)  ,
\]
for $x,y\in p\mathbb{Z}_{p}$, and $\widetilde{S}_{a}(x):=S(a+x)$ for $x\in
p\mathbb{Z}_{p}$.

\begin{notation}
We use the variable $x\in\mathcal{K}$ to denote the state of a network. If the
variable $x$ appears in a function having the decoration $\widetilde{\cdot}$,
like $\widetilde{u}_{a}(x,t)$, then $x\in p\mathbb{Z}_{p}$, i.e. the symbol
$\widetilde{\cdot}$ means that the change of variables of type $x\rightarrow
a+x$ has been performed. The purpose of this notation is to avoid introducing
additional variables to denote the states of the networks.
\end{notation}

With this notation master equation (\ref{Eq_Master_2}) can be rewritten as
\begin{equation}
\frac{d\widetilde{u}_{a}(x,t)}{dt}=\sum\limits_{b\in\mathcal{G}}\text{
}\left\{  \text{ }\int\limits_{p\mathbb{Z}_{p}}\left\{  \widetilde{u}%
_{b}(y,t)-\widetilde{u}_{a}(x,t)\right\}  \widetilde{j}_{a,b}(x\mid
y)dy\right\}  -\widetilde{u}_{a}(x,t)\widetilde{S}_{a}(x)\text{ for }%
a\in\mathcal{G}\text{.} \label{Master_Equation_B}%
\end{equation}
The densities $u_{a}(x,t)$, for $a\in\mathcal{G}$, may exist in times where
the underlying Markov process does not exist due to fact that the network got
trapped in an absorbing state. This type of solutions of the master equation
must be discarded. This restriction makes the study of the master equation a
very involved task. An ultrametric network is a generalization of Markov chain
with $\#\mathcal{G}$\ states, as we will see in the next sections.

\subsection{Some additional remarks}

The problem of determining the functions $\left[  j_{a,b}(x\mid y)\right]
_{a,b\in\mathcal{G}}$, $\left[  S_{a}(x)\right]  _{a\in\mathcal{G}}$ starting
with a free energy landscape of some system is an open problem that requires
to extend the ideas and techniques presented in \ \cite{Becker et al},
\cite{Wales}. This problem is not considered here. On the other hand, the
transition from a state $y\in b+p\mathbb{Z}_{p}$ to a state $x\in
a+p\mathbb{Z}_{p}$ can be perceived as overcoming the energy
barrier\ separating these states. Following \cite{Av-4} , it is natural to
propose using an Arrhenius type relation, to approximate $j_{a,b}\left(  x\mid
y\right)  $, more precisely,
\begin{equation}
j_{a,b}\left(  x\mid y\right)  \sim\exp\left\{  -\frac{\mathbb{U}_{a,b}\left(
x,y\right)  }{kT}\right\}  , \label{6_Energy_Lan}%
\end{equation}
where $\mathbb{U}_{a,b}\left(  x,y\right)  $ is the function describing the
height of the activation barrier for the transition from the state $y\in
b+p\mathbb{Z}_{p}$ to state $x\in a+p\mathbb{Z}_{p}$, $k$ is the Boltzmann
constant and $T$ is the temperature. It is important to emphasize that we are
\ not choosing a reaction coordinate to describe $j_{a,b}\left(  x\mid
y\right)  $. Formula (\ref{6_Energy_Lan})\ establishes a relation between the
structure of the energy landscape $\mathbb{U}_{a,b}\left(  x,y\right)  $ and
the transition function $j_{a,b}\left(  x\mid y\right)  $.

\section{\label{Section_6}Transition networks, fast transition modes, and
Markov chains}

In this section we assume that the restriction of the functions $j(x\mid y)$,
$j(y\mid x)$ to the ball $\left(  a+p\mathbb{Z}_{p}\right)  \times\left(
a+p\mathbb{Z}_{p}\right)  $ have the form%
\[
j(x\mid y)=w_{a}(\left\vert x-y\right\vert _{p})\text{ and \ }j(y\mid
x)=v_{a}(\left\vert x-y\right\vert _{p}),
\]
and that the restriction of the functions $j(x\mid y)$, $j(y\mid x)$ to the
ball $\left(  a+p\mathbb{Z}_{p}\right)  \times\left(  b+p\mathbb{Z}%
_{p}\right)  $, with $a\neq b$, have the form%
\begin{align*}
j_{a,b}(x  &  \mid y)=\lambda_{a,b}\Omega\left(  p\left\vert x-a\right\vert
_{p}\right)  \Omega\left(  p\left\vert y-b\right\vert _{p}\right)  ,\\
j_{b,a}(y  &  \mid x)=\mu_{b,a}\Omega\left(  p\left\vert x-a\right\vert
_{p}\right)  \Omega\left(  p\left\vert y-b\right\vert _{p}\right)  .
\end{align*}
We set%
\[
\lambda_{a,a}:=p\int\limits_{p\mathbb{Z}_{p}}w\left(  \left\vert z\right\vert
_{p}\right)  dz\text{ \ and }\mu_{a,a}:=p\int\limits_{p\mathbb{Z}_{p}}v\left(
\left\vert z\right\vert _{p}\right)  dz.
\]
Now, if $x\in a+p\mathbb{Z}_{p}$ then%
\[%
{\textstyle\int\limits_{\mathcal{K}}}
j(x\mid y)dy=\sum\limits_{\substack{b\in\mathcal{G} \\b\neq a}}p^{-1}%
\lambda_{a,b}+\int\limits_{p\mathbb{Z}_{p}}w_{a}\left(  \left\vert
z\right\vert _{p}\right)  dz=p^{-1}\sum_{a,b\in\mathcal{G}}\lambda
_{a,b}:=p^{-1}\overline{\lambda}_{a}.
\]
Similarly, if $x\in a+p\mathbb{Z}_{p}$ then
\[%
{\textstyle\int\limits_{\mathcal{K}}}
j(y\mid x)dy=\sum\limits_{\substack{b\in\mathcal{G} \\b\neq a}}p^{-1}\mu
_{a,b}+\int\limits_{p\mathbb{Z}_{p}}v_{a}\left(  \left\vert z\right\vert
_{p}\right)  dz=p^{-1}\sum_{a,b\in\mathcal{G}}\mu_{a,b}:=p^{-1}\overline{\mu
}_{a}.
\]
For $x\in a+p\mathbb{Z}_{p}$ we have%
\[
S_{a}(x)=p^{-1}\left(  \overline{\mu}_{a}-\overline{\lambda}_{a}\right)  .
\]
The master equation (\ref{Master_Equation_B}) takes the form%
\begin{gather}
\frac{d\widetilde{u}_{a}(x,t)}{dt}=\sum\limits_{\substack{b\in\mathcal{G}
\\b\neq a}}\lambda_{a,b}\int\limits_{p\mathbb{Z}_{p}}\left\{  \widetilde
{u}_{b}(y,t)-\widetilde{u}_{a}(x,t)\right\}  \Omega\left(  p\left\vert
x\right\vert _{p}\right)  \Omega\left(  p\left\vert y\right\vert _{p}\right)
dy\label{Master_Equation_2}\\
+\int\limits_{p\mathbb{Z}_{p}}\left\{  \widetilde{u}_{a}(y,t)-\widetilde
{u}_{a}(x,t)\right\}  w_{a}(\left\vert x-y\right\vert _{p})dy-p^{-1}\left(
\overline{\mu}_{a}-\overline{\lambda}_{a}\right)  \widetilde{u}_{a}%
(x,t),\nonumber
\end{gather}
for $a\in\mathcal{G}$. With the notation and $Aver$ $\widetilde{u}_{a}%
(\cdot,t)=\int_{p\mathbb{Z}_{p}}\widetilde{u}_{a}(x,t)dx$,
(\ref{Master_Equation_2}) becomes%
\begin{gather}
\frac{d\widetilde{u}_{a}(x,t)}{dt}=\sum\limits_{\substack{b\in\mathcal{G}
\\b\neq a}}\lambda_{a,b}\Omega\left(  p\left\vert x\right\vert _{p}\right)
Aver\text{ }\widetilde{u}_{b}(\cdot,t)-p^{-1}\left(  \overline{\mu}%
_{a}-\lambda_{a,a}\right)  \widetilde{u}_{a}(x,t)\label{Master_Equation_3B}\\
+\int\limits_{p\mathbb{Z}_{p}}\left\{  \widetilde{u}_{a}(y,t)-\widetilde
{u}_{a}(x,t)\right\}  w_{a}(\left\vert x-y\right\vert _{p})dy\text{, \ for
\ }a\in\mathcal{G}.\nonumber
\end{gather}

\begin{remark}
Notice that $\overline{\mu}_{a}=0$, implies that $\mu_{b,a}=0$ for
$a,b\in\mathcal{G}$ and since $\lambda_{a,b}\leq\mu_{b,a}$ for $a,b\in
\mathcal{G}$, by Hypothesis 1, we conclude that $\lambda_{a,b}=0$ for
$a,b\in\mathcal{G}$. For this reason from now on we assume that $\overline
{\mu}_{a}\neq0$ for some $a\in\mathcal{G}$.
\end{remark}

\subsection{The Cauchy problem}

We now consider the Cauchy problem associated with master equation
(\ref{Master_Equation_3B}). We look for a complex-valued solution of the form%
\[
\widetilde{v}_{a}(x,t)=p^{\frac{1}{2}}C_{0}^{\left(  a\right)  }%
(t)\Omega\left(  p\left\vert x\right\vert _{p}\right)  +\sum\limits_{rmj}%
C_{rmj}^{\left(  a\right)  }(t)\Psi_{rmj}\left(  x\right)  \text{ for }%
a\in\mathcal{G},
\]
thus, $\widetilde{u}_{a}(x,t)=\operatorname{Re}\left(  \widetilde{v}%
_{a}(x,t)\right)  $ for $a\in\mathcal{G}$. By substituting in
(\ref{Master_Equation_3B}) one gets%
\begin{equation}
\frac{d}{dt}C_{0}^{\left(  a\right)  }(t)=\sum\limits_{\substack{b\in
\mathcal{G} \\b\neq a}}\lambda_{a,b}C_{0}^{\left(  b\right)  }(t)-p^{-1}%
\left(  \overline{\mu}_{a}-\lambda_{a,a}\right)  C_{0}^{\left(  a\right)
}(t)\text{, for }a\in\mathcal{G}\text{,} \label{Eq_34}%
\end{equation}
i.e. $\frac{d}{dt}\left[  C_{0}^{\left(  a\right)  }(t)\right]  _{a\in
\mathcal{G}}=\Lambda\left[  C_{0}^{\left(  a\right)  }(t)\right]
_{a\in\mathcal{G}}$, with $\Lambda=\left[  \Lambda_{a,b}\right]
_{a,b\in\mathcal{G}}$, where%
\[
\Lambda_{a,b}=\left\{
\begin{array}
[c]{lll}%
\lambda_{a,b} & \text{if} & a\neq b\\
&  & \\
-p^{-1}\left(  \overline{\mu}_{a}-\lambda_{a,a}\right)  & \text{if} & a=b.
\end{array}
\right.
\]
Now, by using (\ref{Spectra}),%
\begin{equation}
\frac{d}{dt}C_{rmj}^{\left(  a\right)  }(t)=\left(  \widehat{w}_{a}\left(
p^{1-r}\right)  -p^{-1}\overline{\mu}_{a}\right)  C_{rmj}^{\left(  a\right)
}(t)\text{, } \label{Eq_35}%
\end{equation}
for $a\in\mathcal{G}$.

Therefore $C_{0}(t)=e^{\Lambda t}C_{0}(0)$,%
\[
C_{rmj}^{\left(  a\right)  }(t)=e^{\left(  \widehat{w}_{a}\left(
p^{1-r}\right)  -p^{-1}\overline{\mu}_{a}\right)  t}C_{rmj}^{\left(  a\right)
}(0)\text{, for }a\in\mathcal{G}\text{,}
\]
and%
\begin{gather}
\left[  u_{a}\left(  x,t\right)  \right]  _{a\in\mathcal{G}}=p^{\frac{1}{2}%
}\left[  C_{0}^{\left(  a\right)  }\left(  t\right)  \Omega\left(  p\left\vert
x-a\right\vert _{p}\right)  \right]  _{a\in\mathcal{G}}\label{Solution_IVP}\\
+\left[
{\textstyle\sum\limits_{rmj}}
e^{\left(  \widehat{w}_{a}\left(  p^{1-r}\right)  -p^{-1}\overline{\mu}%
_{a}\right)  t}\operatorname{Re}\left(  C_{rmj}^{\left(  a\right)  }\left(
0\right)  \Psi_{rmj}\left(  x\right)  \right)  \right]  _{a\in\mathcal{G}%
},\nonumber
\end{gather}
where%
\[
\left[  C_{0}^{\left(  a\right)  }\left(  t\right)  \right]  _{a\in
\mathcal{G}}=e^{\Lambda t}\left[  C_{0}^{\left(  a\right)  }\left(  0\right)
\right]  _{a\in\mathcal{G}}.
\]
Notice that $\widehat{w}_{a}\left(  p^{1-r}\right)  -p^{-1}\overline{\mu}%
_{a}\leq0$ for any $a\in\mathcal{G}$. Notice that $u\left(  x,t\right)  $ is a
real-valued continuous function on $\mathcal{K\times}\left[  0,\infty\right)
$.

\subsection{Existence of fast transition modes}

We set $u^{\left(  0\right)  }\left(  x\right)  =u\left(  x,0\right)  $ for
the initial value of $u\left(  x,t\right)  $, see (\ref{Solution_IVP}). Take
$0\leq u^{\left(  0\right)  }\left(  x\right)  \leq1$. We first notice that
since
\begin{equation}
u(x,t)=\int\limits_{\mathcal{K}}u^{\left(  0\right)  }(y)p_{t}\left(
x,dy\right)  , \label{Solution}%
\end{equation}
then $u(x,t)\geq0$ for any $(x,t)\in\mathcal{K\times}\left[  0,\infty\right)
$. We now recall that if the Markov process $\mathcal{X}$ attached to the
network is conservative, i.e. if%
\[
\int\limits_{\mathcal{K}}1_{\mathcal{K}}(y)p_{t}\left(  x,dy\right)
=1_{\mathcal{K}}(x),
\]
then by (\ref{Solution}), $u(x,t)\leq1$ for any $(x,t)\in\mathcal{K\times
}\left[  0,\infty\right)  $.

We now suppose that the Markov process $\mathfrak{X}$ attached to the network
is not conservative. If the range of $u\left(  x,t\right)  $ contains an
interval of the form $\left[  1,1+\delta\right]  $ for some $\delta>0$, we
define%
\begin{align*}
\tau &  =\tau\left(  u^{\left(  0\right)  }\right)  =\min_{T\in\left(
0,\infty\right)  }\left\{  T;\text{ }u\left(  x_{0},t\right)  \geq1\text{ for
some }x_{0}\in\mathcal{K}\text{, and }t\in\left[  T,T+\epsilon\right]  \text{
for some }\epsilon>0\right\} \\
&  =\min_{a\in\mathcal{G}}\min_{T\in\left(  0,\infty\right)  }\left\{
T;\text{ }u_{a}\left(  x_{0},t\right)  \geq1\text{ for some }x_{0}\in
a+p\mathbb{Z}p\text{, and }t\in\left[  T,T+\epsilon\right]  \text{ for some
}\epsilon>0\right\}  .
\end{align*}

If the range of $0\leq u\left(  x,t\right)  \leq1$ for any $t\geq0$ and any
$x\in\mathcal{K}$, we define%
\[
\tau=\tau\left(  u^{\left(  0\right)  }\right)  =\infty.
\]
If $\tau<\infty$, the network with master equation (\ref{Master_Equation_3B})
moves in the energy landscape, until it gets trapped in some absorbing state
for $t\geq\tau$. Since $u(x,t)$ is a locally constant function in $x$ for each
$t$, see (\ref{Solution_IVP}), the condition $u\left(  x_{0},T\right)  =1$ is
valid in ball around $x_{0}$. Notice that the process does not die because
$x_{0}\in\mathcal{K}$. Furthermore, the solution $u(x,t) $ does not have
probabilistic meaning for $t>\tau\left(  u^{\left(  0\right)  }\right)  $. By
using Theorem \ref{Theorem1}, we obtain the following result.

\begin{theorem}
\label{Theorem2} Assume that the Markov process $\mathfrak{X}$ attached to the
master equation (\ref{Master_Equation_3B}) is not conservative. Take
$u^{\left(  0\right)  }\left(  x\right)  \in$ $\mathcal{C}(\mathcal{K}%
,\mathbb{R}_{+})$ such that $0\leq u^{\left(  0\right)  }\left(  x\right)
\leq1$. Then, the solution $u(x,t)$ of the Cauchy problem attached to the
master equation (\ref{Master_Equation_3B}), with initial datum $u^{\left(
0\right)  }\left(  x\right)  $, is given by (\ref{Solution_IVP}) for
$t\in\left[  0,\tau\left(  u^{\left(  0\right)  }\right)  \right]  $. In the
case $\tau\left(  u^{\left(  0\right)  }\right)  <\infty$, the terms
\begin{gather}
e^{\left(  \widehat{w}_{a}\left(  p^{1-r}\right)  -p^{-1}\overline{\mu}%
_{a}\right)  t}\operatorname{Re}\left(  C_{rmj}^{\left(  a\right)  }\left(
0\right)  \Psi_{rmj}\left(  x\right)  \right)  =\label{Fast modes}\\
e^{\left(  \widehat{w}_{a}\left(  p^{1-r}\right)  -p^{-1}\overline{\mu}%
_{a}\right)  t}\left\{  \operatorname{Re}\left(  C_{rmj}^{\left(  a\right)
}\left(  0\right)  \right)  \cos\left(  \left\{  p^{r-1}jx\right\}
_{p}\right)  -\operatorname{Im}\left(  C_{rmj}^{\left(  a\right)  }\left(
0\right)  \right)  \sin\left(  \left\{  p^{r-1}jx\right\}  _{p}\right)
\right\}  ,\nonumber
\end{gather}
for $a\in\mathcal{G}$, which originate due to the ultrametric structure of the
basins, give rise to fast transitions that drive the network to an absorbing
state (an open compact subset) in a finite time. We call these terms
\textit{fast transition modes}.

Assume that Markov process $\mathfrak{X}$ attached to the master equation
(\ref{Master_Equation_3B}) is conservative. Then, the solution $u(x,t)$ of the
Cauchy problem attached to the master equation (\ref{Master_Equation_3B}),
with initial datum $u^{\left(  0\right)  }\left(  x\right)  $, is given by
(\ref{Solution_IVP}) for $t\in\left[  0,\infty\right)  $.
\end{theorem}

\subsection{Some results about semigroups of matrices}

We now study the connections between the matrix semigroup $e^{t\Lambda}$ and
continuous-time Markov chains.

\begin{definition}
A real matrix $B=\left[  B_{ij}\right]  _{i,j\in I}$, $I=\left\{
1,\ldots,n\right\}  $, is said to be \textit{diagonally dominant}, if
\[
\left\vert B_{ii}\right\vert \geq%
{\displaystyle\sum\limits_{j\neq i}}
\left\vert B_{ij}\right\vert \text{ \ for any }i\in I.
\]
It is \textit{strictly} \textit{diagonally dominant} if
\[
\left\vert B_{ii}\right\vert >%
{\displaystyle\sum\limits_{j\neq i}}
\left\vert B_{ij}\right\vert \text{ \ for any }i\in I.
\]

\end{definition}

\begin{proposition}
[{\cite[Theorem 6.1.10]{Horn and Jhonson}}]\label{Prop1} Let $B=\left[
B_{ij}\right]  _{i,j\in I}$ be a strictly diagonally dominant matrix. Then (i)
$B$ is non singular; (ii) if $B_{ii}>0$ for all $i\in I$, then every
eigenvalue of $B$ has a positive real part; (iii) if $B$ is symmetric and
$B_{ii}>0$ for all $i\in I$, then $B$ is positive definite.
\end{proposition}

\begin{definition}
A square matrix $A=\left[  A_{ij}\right]  _{i,j\in I}$ is called a $Z$-matrix
if the off diagonal elements of $A$ are non positive, i.e. $A_{ij}\leq0$ for
$i\neq j$. In the case in which $A$ is nonsingular, $A$ is called a $M$-matrix
if it is a $Z$-matrix, and its inverse is a nonnegative matrix.
\end{definition}

We denote by $\boldsymbol{1}$ the column vector having all its entries equal
to one. For a column vector $v=\left[  v_{j}\right]  _{j\in I}$, we set
$\left\vert v\right\vert =\max_{j\in I}\left\vert v_{j}\right\vert $. If
$A=\left[  A_{ij}\right]  _{i,j\in I}$, $B=\left[  B_{ij}\right]  _{i,j\in I}$
are matrices, we use the notation $A\leq B$ to mean $A_{ij}\leq B_{ij}$ for
all $i$, $j\in I$.

\begin{definition}
A semigroup of matrices $\left\{  \boldsymbol{Q}_{t}\right\}  _{t\geq0}$ is
called a substochastic Markov semigroup\ if the $\boldsymbol{Q}_{t}$ are
nonnegative matrices and $\boldsymbol{Q}_{t}\boldsymbol{1}\leq\boldsymbol{1}$,
with strict inequality at some coordinate. If $\boldsymbol{Q}_{t}%
\boldsymbol{1}=\boldsymbol{1}$ for all $t$, $\left\{  \boldsymbol{Q}%
_{t}\right\}  _{t\geq0}$ is called a Markov semigroup.
\end{definition}

\begin{proposition}
[{\cite[Theorem 5.34-(ii), Theorem 2.27]{Dellacherie et al}}]\label{Prop2}(i)
$A$ is a \textit{diagonally dominant} $Z$-matrix if and only if
$\boldsymbol{Q}_{t}=e^{-tA}$ is a non-negative semigroup and for all $t\geq0$
it verifies that $\boldsymbol{Q}_{t}\boldsymbol{1}\leq\boldsymbol{1} $.

(ii) If $A$ is nonsingular, then $A$ is a diagonally dominant $M$-matrix if
and only if $\boldsymbol{Q}_{t}=e^{-tA}$ is a substochastic semigroup.
Furthermore, there is a constant $0<\rho<1$ such that for every $v=\left[
v_{j}\right]  _{j\in I}\in\mathbb{R}^{n}$ and $t\geq0$,%
\[
\left\vert \boldsymbol{Q}_{t}v\right\vert \leq\left\vert v\right\vert
\rho^{t-1}.
\]

\end{proposition}

\subsection{Markov chains, ultrametric networks and time scales}

We identify a master equation of type (\ref{Master_Equation_3B}) with a point
$\left(  \lambda_{a,b},\mu_{a,b}\right)  $ in
\[
\mathcal{N}:=\left\{  \left(  \lambda_{a,b},\mu_{a,b}\right)  \in
\mathbb{R}_{+}^{\left(  \#\mathcal{G}\right)  ^{2}}\times\mathbb{R}%
_{+}^{\left(  \#\mathcal{G}\right)  ^{2}};\text{ }\lambda_{a,b}\leq\mu
_{b,a}\text{ for }a,b\in\mathcal{G}\right\}  .
\]
From now on we will identify a network with a point from $\mathcal{N}$. Given
$\mathcal{J}\subset\mathcal{G}$, we set%
\[
\mathcal{N}_{\mathcal{J}}^{\left(  1\right)  }:=\left\{  \left(  \lambda
_{a,b},\mu_{a,b}\right)  \in\mathbb{R}_{+}^{\left(  \#\mathcal{G}\right)
^{2}}\times\mathbb{R}_{+}^{\left(  \#\mathcal{G}\right)  ^{2}};p^{-1}%
\overline{\mu}_{a}+(1-p^{-1})\lambda_{a,a}=\overline{\lambda}_{a}\text{ for
all }a\in\mathcal{J}\right\}  ,
\]%
\[
\mathcal{N}_{\mathcal{J}}^{\left(  2\right)  }:=\left\{  \left(  \lambda
_{a,b},\mu_{a,b}\right)  \in\mathbb{R}_{+}^{\left(  \#\mathcal{G}\right)
^{2}}\times\mathbb{R}_{+}^{\left(  \#\mathcal{G}\right)  ^{2}};p^{-1}%
\overline{\mu}_{a}+(1-p^{-1})\lambda_{a,a}>\overline{\lambda}_{a}\text{ for
all }a\in\mathcal{J}\right\}  ,
\]
with the the convention $\mathcal{N}_{\varnothing}^{\left(  1\right)
}=\mathcal{N}_{\varnothing}^{\left(  2\right)  }=\mathbb{R}_{+}^{\left(
\#\mathcal{G}\right)  ^{2}}\times\mathbb{R}_{+}^{\left(  \#\mathcal{G}\right)
^{2}}$. Now, given%
\begin{equation}
\mathcal{G}_{1},\mathcal{G}_{2}\subset\mathcal{G},\text{with }\mathcal{G}%
_{1}\cap\mathcal{G}_{2}=\varnothing\text{ and }\mathcal{G}_{1}%
{\textstyle\bigsqcup}
\mathcal{G}_{2}=\mathcal{G}, \label{Conditions_G}%
\end{equation}
we set%
\[
\mathcal{N}_{\mathcal{G}_{1},\mathcal{G}_{2}}:=\left\{  \left(  \lambda
_{a,b},\mu_{a,b}\right)  \in\mathbb{R}_{+}^{\left(  \#\mathcal{G}\right)
^{2}}\times\mathbb{R}_{+}^{\left(  \#\mathcal{G}\right)  ^{2}};\text{ }%
\lambda_{a,b}\leq\mu_{b,a}\text{, }a,b\in\mathcal{G}\right\}  \cap
\mathcal{N}_{\mathcal{G}_{1}}^{\left(  1\right)  }\cap\mathcal{N}%
_{\mathcal{G}_{2}}^{\left(  2\right)  }.
\]

\begin{theorem}
\label{Theorem3}Assume Hypothesis 1. We denote by $\left\{  \boldsymbol{P}%
_{t}\right\}  _{t\geq0}$ the Feller semigroup associated with the master
equation (\ref{Master_Equation_3B}). The set of all possible networks
$\mathcal{N}$ is the disjoint union of the sets $\mathcal{N}_{\mathcal{G}%
_{1},\mathcal{G}_{2}}$, where $\mathcal{G}_{1},\mathcal{G}_{2}$ satisfy
(\ref{Conditions_G}). In the set
\[
\mathcal{N}_{\mathcal{G}\text{,}\varnothing}=\left\{
\begin{array}
[c]{r}%
\left(  \lambda_{a,b},\mu_{a,b}\right)  \in\mathbb{R}_{+}^{\left(
\#\mathcal{G}\right)  ^{2}}\times\mathbb{R}_{+}^{\left(  \#\mathcal{G}\right)
^{2}};\lambda_{a,b}\leq\mu_{b,a}\text{, }a,b\in\mathcal{G}\text{ and }\\
\\
p^{-1}\overline{\mu}_{a}+(1-p^{-1})\lambda_{a,a}=\overline{\lambda}_{a}\text{,
}a\in\mathcal{G}%
\end{array}
\right\}  ,
\]
it verifies that $\Lambda\boldsymbol{1}=\boldsymbol{0}$, which implies that
$\left\{  \boldsymbol{P}_{t}\right\}  _{t\geq0}$ is a Markov semigroup. More
precisely, $\Lambda$ is the transition rate matrix (or infinitesimal generator
matrix) of a continuos time Markov chain with $\#\mathcal{G}$\ states.
Furthermore, $\Lambda\boldsymbol{1}\neq\boldsymbol{0}$ in any set
$\mathcal{N}_{\mathcal{G}_{1},\mathcal{G}_{2}}$, with $\mathcal{G}%
_{1}\subsetneqq\mathcal{G}$.

In the set%
\[
\mathcal{N}_{\varnothing,\mathcal{G}}=\left\{
\begin{array}
[c]{r}%
\left(  \lambda_{a,b},\mu_{a,b}\right)  \in\mathbb{R}_{+}^{\left(
\#\mathcal{G}\right)  ^{2}}\times\mathbb{R}_{+}^{\left(  \#\mathcal{G}\right)
^{2}};\lambda_{a,b}\leq\mu_{b,a}\text{, }a,b\in\mathcal{G}\text{ and}\\
\\
\text{ }p^{-1}\overline{\mu}_{a}+(1-p^{-1})\lambda_{a,a}>\overline{\lambda
}_{a}\text{, }a\in\mathcal{G}%
\end{array}
\right\}  ,
\]
under the condition $\overline{\mu}_{a}>\lambda_{a,a}$, for all $a\in
\mathcal{G}$,%
\begin{equation}
\lim_{t\rightarrow\infty}\left\vert e^{t\Lambda}\left[  C_{0}^{\left(
a\right)  }\left(  0\right)  \right]  _{a\in\mathcal{G}}\right\vert =0\text{,}
\label{Eq_limit}%
\end{equation}
which means that the Markov process attached to the network dies at infinite
time. In the set $\mathcal{N}_{\mathcal{G}_{1},\mathcal{G}_{2}}$,with
$\varnothing\varsubsetneqq\mathcal{G}_{1}\varsubsetneqq\mathcal{G}$,
$\varnothing\subseteq\mathcal{G}_{2}\varsubsetneqq\mathcal{G}$, and under the
condition that $\Lambda$ is not singular, $\boldsymbol{P}_{t}=e^{t\Lambda}$ is
a substochastic semigroup if and only if $-\Lambda$ is a $M$-matrix. In this
last case, the limit (\ref{Eq_limit}) holds true.\ 
\end{theorem}

\begin{proof}
We first notice that Hypothesis 1 implies the restrictions $\lambda_{a,b}%
\leq\mu_{b,a}$ for $a,b\in\mathcal{G}$.\ We already established in Theorem
\ref{Theorem1} that the semigroup $\left\{  \boldsymbol{P}_{t}\right\}
_{t\geq0}$ attached to a general master equation of type
(\ref{Eq_Cauchy_problem}) or (\ref{Master_Equation_3B})\ is Feller. In
particular $0\leq\boldsymbol{P}_{t}1_{\mathcal{K}}\leq1$ where $1_{\mathcal{K}%
}(x)=\sum_{_{a\in\mathcal{G}}}\Omega\left(  p\left\vert x-a\right\vert
_{p}\right)  \in\mathcal{C}\left(  \mathcal{K},\mathbb{R}\right)  $ is the
characteristic function of $\mathcal{K}$. We now use $1_{\mathcal{K}}(x)$ as
an initial datum, then by (\ref{Solution_IVP}), with $C_{rmj}^{\left(
a\right)  }\left(  0\right)  =0$ for any $rmj$, we have $\boldsymbol{0}\leq
e^{t\Lambda}\boldsymbol{1}\leq\boldsymbol{1}$ for all $t$, where
$\boldsymbol{0}=\left[  0\right]  _{a\in\mathcal{G}}$. Which means that
$e^{t\Lambda}$ is a nonnegative matrix semigroup and for all $t\geq0$ it
verifies that $e^{t\Lambda}\boldsymbol{1}\leq\boldsymbol{1}$. Now, by
Proposition \ref{Prop2}-(i), $-\Lambda$ is a row diagonally dominant
$Z$-matrix, which implies that
\[
p^{-1}\left\vert \overline{\mu}_{a}-\lambda_{a,a}\right\vert \geq
\sum\limits_{b\in\mathcal{G},\text{ }b\neq a}\lambda_{a,b}\text{, for all
}a\in\mathcal{G}\text{.}
\]
Since $\overline{\mu}_{a}\geq\mu_{a,a}\geq\lambda_{a,a}$ for every
$a\in\mathcal{G}$, We rewrite this condition using the following two
inequalities:%
\begin{equation}
p^{-1}\overline{\mu}_{a}-p^{-1}\lambda_{a,a}=\sum\limits_{b\in\mathcal{G}%
,\text{ }b\neq a}\lambda_{a,b}\Leftrightarrow p^{-1}\overline{\mu}%
_{a}+(1-p^{-1})\lambda_{a,a}=\overline{\lambda}_{a}, \label{Cond_1}%
\end{equation}%
\begin{equation}
p^{-1}\overline{\mu}_{a}-p^{-1}\lambda_{a,a}>\sum\limits_{b\in\mathcal{G}%
,\text{ }b\neq a}\lambda_{a,b}\Leftrightarrow p^{-1}\overline{\mu}%
_{a}+(1-p^{-1})\lambda_{a,a}>\overline{\lambda}_{a}. \label{Cond_2}%
\end{equation}
The existence of the sets $\mathcal{N}_{\mathcal{G}_{1},\mathcal{G}_{2}}$
comes from the fact that $-\Lambda$ is a row diagonally dominant $Z$-matrix,
and the above discussion. Notice that in $\mathcal{N}_{\mathcal{G}%
\text{,}\varnothing}$, it verifies that $\Lambda\boldsymbol{1}=\boldsymbol{0}%
$, which means that $\Lambda$ is the infinitesimal generator of a conservative
Markov process, see e.g. \cite[Lemma 2.3]{Dynkin}, more precisely of a
continuous time Markov chain with $\#\mathcal{G}$ states.

In the set $\mathcal{N}_{\varnothing,\mathcal{G}}$, $-\Lambda$ is nonsingular,
and under the condition $\overline{\mu}_{a}>\lambda_{a,a}$, for all
$a\in\mathcal{G}$, all the elements in the diagonal of $-\Lambda$ are
positive, then by Proposition \ref{Prop1}-(ii), all the eigenvalues of
$\Lambda$ have positive real part, which implies the limit (\ref{Eq_limit}).
In the set $\mathcal{N}_{\mathcal{G}_{1},\mathcal{G}_{2}}$,with $\varnothing
\varsubsetneqq\mathcal{G}_{1}\varsubsetneqq\mathcal{G}$, $\varnothing
\subseteq\mathcal{G}_{2}\varsubsetneqq\mathcal{G}$, and under the condition
$\Lambda$ is not singular, $\boldsymbol{P}_{t}=e^{t\Lambda}$ is a
substochastic semigroup if and only if $-\Lambda$ is a $M$-matrix, cf.
Proposition \ref{Prop2}-(ii).
\end{proof}

\begin{theorem}
\label{Theorem4}(i) Assume that the Markov process $\mathfrak{X}$ attached to
the network is conservative. For any $u^{\left(  0\right)  }\left(  x\right)
\in$ $\mathcal{C}(\mathcal{K},\mathbb{R}_{+})$ such that $0\leq u^{\left(
0\right)  }\left(  x\right)  \leq1$,it verifies that $\tau\left(  u^{\left(
0\right)  }\right)  =\infty$, and the solution $u(x,t)\in\left[  0,1\right]  $
of the Cauchy problem attached to the master equation
(\ref{Master_Equation_3B}), with initial datum $u^{\left(  0\right)  }\left(
x\right)  $, is given by (\ref{Solution_IVP}) for $t\in\left[  0,\infty
\right)  $. Furthermore,
\[
\lim_{t\rightarrow\infty}\left[  u_{a}\left(  x,t\right)  \right]
_{a\in\mathcal{G}}=p^{\frac{1}{2}}\lim_{t\rightarrow\infty}e^{\Lambda
t}\left[  C_{0}^{\left(  a\right)  }\left(  0\right)  \right]  _{a\in
\mathcal{G}},
\]
where $\Lambda$ is the infinitesimal generator of a continuous-time Markov
chain with $\#\mathcal{G}$ states.

(ii) Assume that the Markov process $\mathfrak{X}$ attached to the network is
not conservative. If $\tau\left(  u^{\left(  0\right)  }\right)  =\infty$,
with $u^{\left(  0\right)  }$ as in part (i), the solution $u(x,t)\in\left[
0,1\right]  $ of the Cauchy problem attached to the master equation
(\ref{Master_Equation_3B}), with initial datum $u^{\left(  0\right)  }\left(
x\right)  $, is given by (\ref{Solution_IVP}) for $t\in\left[  0,\infty
\right)  $. Furthermore,
\[
\lim_{t\rightarrow\infty}\left[  u_{a}\left(  x,t\right)  \right]
_{a\in\mathcal{G}}=p^{\frac{1}{2}}\lim_{t\rightarrow\infty}e^{\Lambda
t}\left[  C_{0}^{\left(  a\right)  }\left(  0\right)  \right]  _{a\in
\mathcal{G}}=0,
\]
where $\left\{  e^{\Lambda t}\right\}  _{t\geq0}$ is a substochastic
semigroup. By introducing a terminal state $\Delta%
{\textstyle\bigsqcup}
\mathcal{G}$, $\ \Lambda$ becomes the infinitesimal generator of a killed
continuous-time Markov chain with $1+\#\mathcal{G}$ states.

(iii) Assume that the Markov process $\mathfrak{X}$ attached to the network is
not conservative. If $\tau\left(  u^{\left(  0\right)  }\right)  <\infty$,
with $u^{\left(  0\right)  }$ as in part (i), the solution $u(x,t)\in\left[
0,1\right]  $ of the Cauchy problem attached to the master equation
(\ref{Master_Equation_3B}), with initial datum $u^{\left(  0\right)  }\left(
x\right)  $, is given by (\ref{Solution_IVP}) for $t\in\left[  0,\tau\left(
u^{\left(  0\right)  }\right)  \right]  $. Furthermore, at the time
$t=\tau\left(  u^{\left(  0\right)  }\right)  $ the network reaches an
absorbing state.
\end{theorem}

\begin{proof}
The result follows from Theorems \ref{Theorem2}, \ref{Theorem3}, by using that
the restriction of $e^{t\boldsymbol{J}}$ to the space $%
{\textstyle\bigoplus\nolimits_{a\in\mathcal{G}}}
$ $\mathbb{R}$ $\Omega\left(  p\left\vert x-a\right\vert \right)  $ is
$e^{t\Lambda}$.
\end{proof}

\begin{remark}
Each term $e^{\left(  \widehat{w}_{a}\left(  p^{1-r}\right)  -p^{-1}%
\overline{\mu}_{a}\right)  t}\operatorname{Re}\left(  C_{rmj}^{\left(
a\right)  }\left(  0\right)  \Psi_{rmj}\left(  x\right)  \right)  $
corresponds to the response of a second order type system, which has a time
constant%
\[
\sigma_{a,r}=\frac{4}{-\left(  \widehat{w}_{a}\left(  p^{1-r}\right)
-p^{-1}\overline{\mu}_{a}\right)  }\text{. }%
\]
Based on this fact, we argue that the dynamics of each basin is determined by
the interaction of a countable number of second order type systems, each of
them with particular time constant $\sigma_{a,r}$, $a\in\mathcal{G}$, $r\leq
1$. The dynamics between basins is controlled by a Markov chain determined by
the matrix $\Lambda$. The network has a countable number of time scales:
$\left\{  \sigma_{a,r}\right\}  $, and $\frac{-4}{\mu}$, where $\mu$ runs
through the eigenvalues of $\Lambda$.
\end{remark}

\section{\label{Section_7}Binary Networks and a toy model for protein folding}

\subsection{A class of linear systems in the plane}

For $\alpha$, $\beta$, $\gamma>0$, with $\beta\geq\alpha$, $\gamma\geq\alpha$,
we set%
\[
\frac{d}{dt}\left[
\begin{array}
[c]{c}%
z_{1}\left(  t\right) \\
z_{2}\left(  t\right)
\end{array}
\right]  =\Gamma\left[
\begin{array}
[c]{c}%
z_{1}\left(  t\right) \\
z_{2}\left(  t\right)
\end{array}
\right]  \text{, where \ }\Gamma:=\left[
\begin{array}
[c]{cc}%
-\beta & \alpha\\
\alpha & -\gamma
\end{array}
\right]  .
\]
Take $A:=\sqrt{4\alpha^{2}+\left(  \beta-\gamma\right)  ^{2}}>0$. The
eigenvalues of $\Gamma$ are negative real numbers:%
\[
\frac{-1}{2}\left(  \beta+\gamma+A\right)  ,\text{ \ \ }\frac{1}{2}\left(
A-\gamma-\beta\right)  .
\]
On the other hand,
\begin{gather*}
e^{t\Gamma}=e^{t\left(  -\frac{1}{2}\beta-\frac{1}{2}\gamma+\frac{1}%
{2}A\right)  }\times\\
\left[
\begin{array}
[c]{cc}%
\frac{1}{2A}\left(  \beta-\gamma+A\right)  e^{-tA}+\frac{1}{2A}\left(
-\beta+\gamma+A\right)  & \frac{1}{A}\alpha-\frac{1}{A}\alpha e^{-tA}\\
& \\
-\frac{\alpha}{A}e^{-tA}+\frac{\alpha}{A} & \frac{1}{2A}\left(  \beta
-\gamma+A\right)  -\frac{1}{2A}\left(  \beta-\gamma-A\right)  e^{-tA}%
\end{array}
\right]  .
\end{gather*}
where . Now for $t>0$ and $A$ sufficiently large, we use the approximation:%
\[
e^{t\Gamma}\simeq e^{t\left(  -\frac{1}{2}\beta-\frac{1}{2}\gamma+\frac{1}%
{2}A\right)  }\left[
\begin{array}
[c]{cc}%
\frac{1}{2A}\left(  -\beta+\gamma+A\right)  & \frac{1}{A}\alpha\\
& \\
\frac{\alpha}{A} & \frac{1}{2A}\left(  \beta-\gamma+A\right)
\end{array}
\right]  .
\]

\subsection{A class of binary networks}

We now consider a network with two basins, denoted as $U$, $N$. The master
equation for this network is%
\begin{equation}
\left\{
\begin{array}
[c]{l}%
\frac{\partial\widetilde{v}_{U}(x,t)}{\partial t}=\lambda_{U,N}\text{ }%
\Omega\left(  p\left\vert x\right\vert _{p}\right)  Aver\widetilde{v}%
_{N}\left(  \cdot,t\right)  -p^{-1}\left(  \overline{\mu}_{U}-\lambda
_{U,U}\right)  \widetilde{v}_{U}(x,t)dy\\
\\
\multicolumn{1}{r}{+\int\limits_{p\mathbb{Z}_{p}}\left\{  \widetilde{v}%
_{U}(y,t)-\widetilde{v}_{U}(x,t)\right\}  w_{U}(\left\vert x-y\right\vert
_{p})dy;}\\
\multicolumn{1}{r}{}\\
\multicolumn{1}{r}{\frac{\partial\widetilde{v}_{N}(x,t)}{\partial t}%
=\lambda_{N,U}\text{ }\Omega\left(  p\left\vert x\right\vert _{p}\right)
Aver\widetilde{v}_{U}\left(  \cdot,t\right)  -p^{-1}\left(  \overline{\mu}%
_{N}-\lambda_{N,N}\right)  \text{ }\widetilde{v}_{N}\left(  x,t\right) }\\
\\
\multicolumn{1}{r}{+\int\limits_{p\mathbb{Z}_{p}}\left\{  \widetilde{v}%
_{N}(y,t)-\widetilde{v}_{N}(x,t)\right\}  w_{N}(\left\vert x-y\right\vert
_{p})dy,}%
\end{array}
\right.  \label{Master_Equation_5}%
\end{equation}
where $\overline{\lambda}_{U}=\lambda_{U,U}+\lambda_{U,N}$, $\overline
{\lambda}_{N}=\lambda_{N,N}+\lambda_{N,U}$, $\overline{\mu}_{U}=\mu_{U,U}%
+\mu_{U,N}$, $\overline{\mu}_{N}=\mu_{N,N}+\mu_{N,U}$, and%
\[
S_{a}\left(  x\right)  =p^{-1}\left(  \overline{\mu}_{a}-\overline{\lambda
}_{a}\right)  \text{, }a\in\left\{  N,U\right\}  \text{ and }\lambda_{a,b}%
\leq\mu_{b,a}\text{, }a,b\in\left\{  N,U\right\}  .
\]
For convenience we use $\widetilde{v}(x,t)$ instead of $\widetilde{u}(x,t)$.
The matrix $\Lambda$\ is%
\[
\Lambda=\left[
\begin{array}
[c]{cc}%
-p^{-1}\left(  \overline{\mu}_{U}-\lambda_{U,U}\right)  & \lambda_{U,N}\\
& \\
\lambda_{N,U} & -p^{-1}\left(  \overline{\mu}_{N}-\lambda_{N,N}\right)
\end{array}
\right]  ,
\]
we assume that $\lambda_{N,U}=\lambda_{U,N}$. Using (\ref{Solution_IVP}), \ we
now pick a solution of the form%
\[
\left[
\begin{array}
[c]{c}%
\widetilde{v}_{U}\left(  x,t\right) \\
\\
\widetilde{v}_{N}\left(  x,t\right)
\end{array}
\right]  =p^{\frac{1}{2}}\left[
\begin{array}
[c]{c}%
C_{0}^{\left(  U\right)  }\left(  t\right)  \Omega\left(  p\left\vert
x\right\vert _{p}\right) \\
\\
C_{0}^{\left(  N\right)  }\left(  t\right)  \Omega\left(  p\left\vert
x\right\vert _{p}\right)
\end{array}
\right]  +\left[
\begin{array}
[c]{c}%
0\\
\\
C_{r}^{\left(  N\right)  }e^{\left(  \widehat{w}_{N}\left(  p^{1-r}\right)
-p^{-1}\overline{\mu}_{N}\right)  t}\cos\left(  \left\{  p^{-r-1}x\right\}
_{p}\right)
\end{array}
\right]  ,
\]
where $r>1$, and%
\[
\left[
\begin{array}
[c]{c}%
C_{0}^{\left(  U\right)  }\left(  t\right) \\
\\
C_{0}^{\left(  N\right)  }\left(  t\right)
\end{array}
\right]  =e^{t\Lambda}\left[
\begin{array}
[c]{c}%
p^{\frac{-1}{2}}\\
\\
0
\end{array}
\right]  .
\]
We apply the results of the previous section with $\beta=p^{-1}\left(
\overline{\mu}_{U}-\lambda_{U,U}\right)  $, $\alpha=\lambda_{N,U}%
=\lambda_{U,N}$, and $\gamma=p^{-1}\left(  \overline{\mu}_{N}-\lambda
_{N,N}\right)  $:
\begin{align*}
\left[
\begin{array}
[c]{c}%
C_{0}^{\left(  U\right)  }\left(  t\right) \\
\\
C_{0}^{\left(  N\right)  }\left(  t\right)
\end{array}
\right]   &  =e^{t\Lambda}\left[
\begin{array}
[c]{c}%
p^{\frac{-1}{2}}\\
\\
0
\end{array}
\right]  \simeq\\
&  e^{t\left(  -\frac{1}{2}\beta-\frac{1}{2}\gamma+\frac{1}{2}A\right)
}\left[
\begin{array}
[c]{cc}%
\frac{1}{2A}\left(  -\beta+\gamma+A\right)  & \frac{1}{A}\alpha\\
& \\
\frac{\alpha}{A} & \frac{1}{2A}\left(  \beta-\gamma+A\right)
\end{array}
\right]  \left[
\begin{array}
[c]{c}%
p^{\frac{-1}{2}}\\
\\
0
\end{array}
\right] \\
&  =e^{t\left(  -\frac{1}{2}\beta-\frac{1}{2}\gamma+\frac{1}{2}A\right)
}\left[
\begin{array}
[c]{c}%
\frac{p^{\frac{-1}{2}}}{2A}\left(  A-\beta+\gamma\right) \\
\\
\frac{p^{\frac{-1}{2}}}{A}\alpha
\end{array}
\right]  .
\end{align*}
Which implies that%
\begin{equation}
\left[
\begin{array}
[c]{c}%
\widetilde{v}_{U}\left(  x,0\right) \\
\\
\widetilde{v}_{N}\left(  x,0\right)
\end{array}
\right]  =\left[
\begin{array}
[c]{c}%
\frac{1}{2A}\left(  A-\beta+\gamma\right)  \Omega\left(  p\left\vert
x\right\vert _{p}\right) \\
\\
\frac{\alpha}{A}\Omega\left(  p\left\vert x\right\vert _{p}\right)
\end{array}
\right]  +\left[
\begin{array}
[c]{c}%
0\\
\\
C_{r}^{\left(  N\right)  }\cos\left(  \left\{  p^{-1-r}x\right\}  _{p}\right)
\end{array}
\right]  . \label{IVP_II}%
\end{equation}
Notice that $A-\beta+\gamma>0$.

The function $\widetilde{v}_{U}\left(  x,0\right)  $ gives the concentration
(or population) at time zero for the basin $U$. In particular, the average
concentration of basin $U$ at time zero is $\int\widetilde{v}_{U}\left(
x,0\right)  dx=\frac{1}{2Ap}\left(  A-\beta+\gamma\right)  $. The function
$\widetilde{v}_{N}\left(  x,0\right)  $ gives the concentration at time zero
for the basin $N$. The average concentration for basin $N$ at time zero is
$\int\widetilde{v}_{N}\left(  x,0\right)  dx=\frac{\alpha}{Ap}\simeq0$. Notice
that the functions $\widetilde{v}_{U}\left(  x,0\right)  $, $\widetilde{v}%
_{N}\left(  x,0\right)  $ cannot be interpreted as probability densities since
they do not satisfy $\int\widetilde{v}_{U}\left(  x,0\right)  dx+\int
\widetilde{v}_{N}\left(  x,0\right)  dx=1$. We do not use this last
normalization, since the calculation without imposing this normalization is simpler.

Our goal is to show that the network (\ref{Master_Equation_5}), with initial
density (\ref{IVP_II}) gets trapped in the basin $N$ in a finite time.

\subsection{Computation of $T(v\left(  x,0\right)  )$}

In order to compute $T(v\left(  x,0\right)  )$, we require the conditions:%
\[
0\leq\frac{1}{2A}\left(  A-\beta+\gamma\right)  \leq1\text{ and }0\leq
\frac{\alpha}{A}+C_{r}^{\left(  N\right)  }\cos\left(  \left\{  p^{-1-r}%
x\right\}  _{p}\right)  \leq1,
\]
for every $x\in p\mathbb{Z}p$. The first condition is true. The last condition
is satisfied if $0\leq\frac{\alpha}{A}+C_{r}^{\left(  N\right)  }\leq1$, with
$C_{r}^{\left(  N\right)  }>0$. Indeed, for $x=x_{1}p+x_{2}p_{2}+\cdots
+x_{r}p^{r}+\cdots\in p\mathbb{Z}p$,
\[
\left\{  p^{-1-r}x\right\}  _{p}=\frac{x_{1}}{p^{r}}+\frac{x_{2}}{p^{r-1}%
}+\cdots+x_{r}p^{r},
\]
and since the $x_{i}\in\left\{  0,\ldots,p-1\right\}  $, we have
\begin{equation}
\left\{  p^{-1-r}x\right\}  _{p}\in\left[  0,1-p^{-r}\right]  ,
\label{Condition_r}%
\end{equation}
and thus $\cos\left(  \left\{  p^{-1-r}x\right\}  _{p}\right)  >0$.

We now look for solutions of the condition:%
\begin{equation}
e^{t\left(  -\frac{1}{2}\beta-\frac{1}{2}\gamma+\frac{1}{2}A\right)  }%
\frac{\alpha}{A}+C_{r}^{\left(  N\right)  }e^{\left(  \widehat{w}_{N}\left(
p^{1-r}\right)  -p^{-1}\overline{\mu}_{N}\right)  t}\cos\left(  \left\{
p^{-1-r}x\right\}  _{p}\right)  \geq1. \label{Condition_1}%
\end{equation}
We set
\[
\sigma:=\max\left\{  -\frac{1}{2}\beta-\frac{1}{2}\gamma+\frac{1}{2}%
A,\widehat{w}_{N}\left(  p^{1-r}\right)  -p^{-1}\overline{\mu}_{N}\right\}
<0,
\]
then%
\begin{align*}
e^{t\sigma}\left(  \frac{\alpha}{A}+C_{r}^{\left(  N\right)  }\cos\left(
\left\{  p^{-1-r}x\right\}  _{p}\right)  \right)   &  \geq\\
e^{t\left(  -\frac{1}{2}\beta-\frac{1}{2}\gamma+\frac{1}{2}A\right)  }%
\frac{\alpha}{A}+C_{r}^{\left(  N\right)  }e^{\left(  \widehat{w}_{N}\left(
p^{1-r}\right)  -p^{-1}\overline{\mu}_{N}\right)  t}\cos\left(  \left\{
p^{-1-r}x\right\}  _{p}\right)   &  \geq1,
\end{align*}
and in order to find solutions for condition (\ref{Condition_1}), we can find
solutions for
\[
e^{t\sigma}\left(  \frac{\alpha}{A}+C_{r}^{\left(  N\right)  }\cos\left(
\left\{  p^{-1-r}x\right\}  _{p}\right)  \right)  \geq1\text{, }
\]
i.e.,
\[
\cos\left(  \left\{  p^{-1-r}x\right\}  _{p}\right)  \geq\frac{1}%
{C_{r}^{\left(  N\right)  }}\left(  e^{-t\sigma}-\frac{\alpha}{A}\right)  .
\]
Now%
\[
\frac{1}{C_{r}^{\left(  N\right)  }}\left(  e^{-T\sigma}-\frac{\alpha}%
{A}\right)  =1\Leftrightarrow T=\frac{1}{-\sigma}\ln\left(  C_{r}^{\left(
N\right)  }+\frac{\alpha}{A}\right)  .
\]
Then, by using the fact that $e^{-t\sigma}-\frac{\alpha}{A}$ is an increasing
function in $\left[  0,\infty\right)  $, given $\delta>0$ sufficiently small,
there exists $\epsilon=\epsilon\left(  \delta\right)  $ such that $\frac
{1}{C_{r}^{\left(  N\right)  }}\left(  e^{-t\sigma}-\frac{\alpha}{A}\right)
\in\left[  1-\epsilon,1\right]  $ for $t\in\left[  T-\delta,T\right]  $. We
now look for solutions of
\begin{equation}
\cos\left(  \left\{  p^{-r-1}x\right\}  _{p}\right)  \in\left[  1-\epsilon
,1\right]  \Leftrightarrow\left\{  p^{-r-1}x\right\}  _{p}\in\left[
0,\epsilon^{^{\prime}}\right]  . \label{Condition_epsilon}%
\end{equation}
for some $\epsilon^{^{\prime}}=\epsilon^{^{\prime}}\left(  \epsilon\right)  $.
Taking $-r$ is sufficiently large and using (\ref{Condition_r}), we conclude
that $\left\{  p^{-r-1}x\right\}  _{p}\in\left[  0,1-p^{-r}\right]
\subset\left[  0,\epsilon^{^{\prime}}\right]  $, and consequently, there
exists $x_{0}$ satisfying (\ref{Condition_epsilon}). We now use that $\left\{
p^{r-1}x\right\}  _{p}$ is locally constant and that $\left\{  0\right\}
_{p}=0 $ to obtain a ball $B_{l}(x_{0})$ such that $\left\{  p^{r-1}x\right\}
_{p}=0$ for any $x\in B_{l}(x_{0})$.

\subsection{Discussion}

A protein is a self-organized system which is able to lower its entropy when
it is in contact with a heat reservoir at appropriate temperature. The protein
folding consists of a sequence of random motions in funnel-like configuration
space from high entropy states to a state of minimum entropy (a native state).
However, there is still the problem (the so-called Levinthal paradox), which
is to know if the random transitions may get trapped in an intermediate state,
and then the sequence of the transitions will not end in a state of minimum
entropy. On the other hand, it is well accepted that if the temperature is too
high the protein will not fold. It is also possible that the protein will not
fold correctly if the temperature is too low. We assume a range of
temperatures where the protein folding occurs. In our toy model the energy
landscape of the protein is divided into two regions $U$ (unfolded state), $N$
(native state). At the time zero the concentration in the basins $U$, $N$ is
given by (\ref{IVP_II}). The average concentration of the basin $N$ is very
close to zero. The time evolution of the concentrations $\widetilde{v}%
_{U}(x,t)$, $\widetilde{v}_{N}(x,t)$ is controlled by the master equation
(\ref{Master_Equation_5}). The solution of the Cauchy problem attached to the
master equation (\ref{Master_Equation_5}), with the initial datum
(\ref{IVP_II}), admits a fast mode $C_{r}^{\left(  N\right)  }e^{\left(
\widehat{w}_{N}\left(  p^{1-r}\right)  -p^{-1}\overline{\mu}_{N}\right)
t}\cos\left(  \left\{  p^{-1-r}x\right\}  _{p}\right)  $, which drives the
network into an absorbing sate inside the basin $N$ at the time
\[
\tau=\frac{\ln\left(  C_{r}^{\left(  N\right)  }+\frac{\alpha}{A}\right)
}{\min\left\{  \frac{1}{2}\beta+\frac{1}{2}\gamma-\frac{1}{2}A,p^{-1}%
\overline{\mu}_{N}-\widehat{w}_{N}\left(  p^{1-r}\right)  \right\}  }\text{.}%
\]
The system has two time scales (or two time constants), under the hypothesis
that $A$ is sufficiently large, the time constants are given by%
\[
\frac{1}{\frac{1}{2}\beta+\frac{1}{2}\gamma-\frac{1}{2}A}\text{, \ \ \ }%
\frac{1}{p^{-1}\overline{\mu}_{N}-\widehat{w}_{N}\left(  p^{1-r}\right)
}\text{.}%
\]
The first one is associated with the matrix $\Lambda$ ( the underlying Markov
chain), and the second with the fast mode. Here we do not use the constant $4$
that appears in the definition of the time constant of second order systems.
Finally, following \cite{Av-4}-\cite{Av-5} the dependency on the temperature
may be introduced by taking $\widehat{w}_{N}\left(  p^{1-r}\right)
=\widehat{w}_{N}\left(  p^{1-r},kT\right)  $.

\end{document}